\documentclass[a4paper,11pt]{article}

\usepackage[backend=biber, style=apa, uniquename=false]{biblatex} 
\usepackage[colorlinks=true, allcolors=blue]{hyperref}  
\usepackage{graphicx}
\usepackage{multicol,multirow}
\usepackage{amsmath,amssymb,amsfonts}
\usepackage{mathrsfs}
\usepackage{amsthm}
\usepackage{rotating}
\usepackage{appendix}
\usepackage{booktabs}
\usepackage{ifpdf}
\usepackage[T1]{fontenc}
\usepackage{newtxtext}
\usepackage{newtxmath}
\usepackage{setspace}
\usepackage{textcomp}
\usepackage[english]{babel} 
\usepackage{xcolor}
\usepackage{lipsum}
\usepackage{listings}
\usepackage{geometry}
\usepackage{bm}     
\usepackage{array}
\usepackage{multirow}
\usepackage{comment}
\usepackage{threeparttable}
\usepackage{listings}
\usepackage{makecell}
\usepackage{colortbl}
\usepackage{mathtools}
\usepackage{wrapfig}
\usepackage{ascmac}
\usepackage{fancybox}
\usepackage{algorithm}
\usepackage{authblk}
\usepackage{algpseudocode}
\usepackage{url}
\definecolor{mygray}{rgb}{0.9,0.9,0.9}

\DeclareMathOperator*{\argmin}{argmin}

\DeclareMathOperator{\E}{\mathbb{E}}

\newcommand{\ctext}[1]{\raise0.2ex\hbox{\textcircled{\scriptsize{#1}}}}

\newtheorem{theorem}{Theorem}[section]

\theoremstyle{definition}

\numberwithin{equation}{section}

\title{Robust Estimation of Item Parameters via Divergence Measures in Item Response Theory}
\author{Yuki Itaya\thanks{Graduate School of Science and Technology, Keio University, Japan}, \hspace{3mm}
        Kenichi Hayashi\thanks{Department of Mathematics, Keio University, Japan}}

\begin{document}
\maketitle

\small
\noindent \textbf{Correspondence}\\ 
Yuki Itaya, Graduate School of Science and Technology, Keio University \\
3-14-1 Hiyoshi, Kohoku, Yokohama, Kanagawa, 223-0061, Japan \\
Email: \url{yuki_0929@keio.jp} \\

\noindent \textbf{Acknowledgments} \\
I am deeply grateful to Associate Professor Kenichi Hayashi, my advisor and collaborator, for his invaluable guidance and support. This work was supported by the Japan Society for the Promotion of Science under Grant Numbers 23K11013. \\

\noindent \textbf{Data Availability Statement}\\
The code for estimating item parameters using the proposed methods is available at \url{https://github.com/yukiitaya/Robust_IRT}. \\


\normalsize
\begin{abstract}
Marginal maximum likelihood estimation (MMLE) in item response theory (IRT) is highly sensitive to aberrant responses, such as careless answering and random guessing, which can reduce estimation accuracy. To address this issue, this study introduces robust estimation methods for item parameters in IRT. Instead of empirically minimizing Kullback--Leibler divergence as in MMLE, the proposed approach minimizes the objective functions based on robust divergences, specifically density power divergence and $\gamma$-divergence. The resulting estimators are statistically consistent and asymptotically normal under appropriate regularity conditions. Furthermore, they offer a flexible trade-off between robustness and efficiency through hyperparameter tuning, forming a generalized estimation framework encompassing MMLE as a special case. To evaluate the effectiveness of the proposed methods, we conducted simulation experiments under various conditions, including scenarios with aberrant responses. The results demonstrated that the proposed methods surpassed existing ones in performance across various conditions. Moreover, numerical analysis of influence functions verified that increasing the hyperparameters effectively suppressed the impact of responses with low occurrence probabilities, which are potentially aberrant. These findings highlight that the proposed approach offers a robust alternative to MMLE, significantly enhancing measurement accuracy in testing and survey contexts prone to aberrant responses.
\end{abstract}

\newpage
\normalsize
\section{Introduction}\label{Introduction}
Assessing latent traits---such as cognitive abilities or personality characteristics---through tests and surveys remains fundamental across fields including education, psychology, and marketing (American Educational Research Association et al., \cite{aera2014}). A key challenge in this process lies in ensuring that measurement instruments capture these traits accurately and consistently, while mitigating distortions caused by unique sample attributes or specific test contexts.

Item response theory (IRT) provides a powerful framework for addressing these concerns by modeling item responses as functions of latent traits (Embretson \& Reise, \cite{embretson2000}). By estimating item-level parameters such as difficulty and discrimination, IRT offers refined, interpretable measures of individual differences. Furthermore, its capacity to maintain parameter invariance across diverse groups ensures consistent score interpretations, enhancing comparability and making it especially valuable for large-scale assessments and adaptive testing. Consequently, IRT has become a mainstay in fields where precise and reliable measurement is paramount.

While IRT has notable advantages, a growing body of research indicates that aberrant responses---patterns of responding that deviate significantly from model expectations---can bias parameter estimates and reduce accuracy (van der Linden \& Hambleton, \cite{vanderlinden1997}). These irregularities may arise from various factors, including test-taker inattention, low motivation, and differential item functioning (DIF) in certain items, which may systematically benefit some groups while disadvantaging others, ultimately compromising the validity of the assessment. If such aberrant responses are not considered, they can substantially compromise the validity and reliability of IRT analyses. Consequently, addressing such responses is crucial for preserving the integrity of IRT-based analysis of tests.

One common approach to handling aberrant responses is applying person-fit statistics, which quantify discrepancies between observed responses and expected response patterns under the model assumptions, and removing respondents identified as outliers (Meijer \& Sijtsma, \cite{meijer2001}; \"Ozt\"urk \& Karabatsos, \cite{ozturk2017}). An alternative strategy is to weight participants' data based on these person-fit statistics and then maximize a weighted marginal likelihood (Hong \& Cheng, \cite{hong2019}). Both approaches enhance the robustness of parameter estimates by mitigating the effects of responses that deviate from expected patterns. However, removing responses or weighting participants entails a risk of losing valuable information from the original dataset. Additionally, the challenge of ``double usage''---employing the same dataset for both weight calculation and parameter estimation---raises concerns about potential biases and credibility of subsequent inferences (Johnson \& Albert, \cite{johnson2006}; Michaelides, \cite{michaelides2010}).

Another line of research incorporates outliers or atypical respondent behavior directly into the measurement model, enabling their effects to be accounted for during parameter estimation. For example, Shao et al.~(\cite{shao2016}) and Wang et al.~(\cite{wang2018}) modeled behaviors such as test speededness,\footnote{Referring to the influence of time constraints on performance when response speed is not the target construct of measurement.} rapid guessing and cheating as part of the measurement process, facilitating their detection and correction during parameter estimation. In psychology, B\"ockenholt (\cite{boeckenholt2017}) and Falk and Cai (\cite{falk2016}) treated aberrant responses as ``response styles,'' distinguishing them from the focal traits under measurement. B\"ockenholt (\cite{boeckenholt2017}) employed an IR-tree model to isolate Likert-scale response styles from main constructs, whereas Falk and Cai (\cite{falk2016}) introduced a flexible multidimensional IRT approach accommodating multiple response styles simultaneously. Rather than viewing response styles as mere random noise, these methods treat them as essential parts of the measurement process, thereby enhancing validity by addressing distortions caused by atypical responding.

\"Ozt\"urk and Karabatsos (\cite{ozturk2017}) further advanced this line of research by introducing a robust Bayesian IRT model that automatically identifies outliers without discarding any observations and subsequently down-weights their influence. In particular, they incorporate latent parameters that indicate whether each response is an outlier, thereby estimating both the outlier status and the item parameters within a single framework. While this integrated approach obviates the need for separately detecting outliers, it poses additional challenges. First, the number of estimable parameters increases, expanding the dimensionality of the parameter space and complicating the estimation procedure. Second, employing Markov chain Monte Carlo (MCMC) methods heightens the computational burden. Moreover, model outcomes can be sensitive to the choice of prior distributions and other modeling decisions, necessitating careful consideration of these factors.

Methods incorporating aberrant responses into the model often require assumptions about the mechanisms driving such responses. In practice, however, multiple respondent characteristics and complex cognitive processes frequently interact, and the presumed mechanisms may not always capture the true one (Meade \& Craig, \cite{meade2012}). Given these considerations, both methods that detect outliers statistically and remove or reweight them and methods that incorporate these anomalous responses directly into the model face inherent limitations. In this study, we propose an alternative perspective: making the estimation process itself robust. Specifically, we focus on IRT parameter estimation methods that employ robust divergence measures, focusing on density power divergence (DPD) introduced in Basu et al.~(\cite{basu1998}) and $\gamma$-divergence from Fujisawa and Eguchi (\cite{fujisawa2008}). Utilizing these divergence makes it feasible to reduce estimation bias in the presence of aberrant responses without discarding potentially informative data. To the best of our knowledge, no studies have yet explored robust estimation approaches via divergence measures under IRT models.

The remainder of this paper is organized as follows. In the rest of this section, we provide an overview of the IRT model and briefly discuss conventional estimation methods for item parameters. Section \ref{robust estimation} introduces two robust divergences---DPD and $\gamma$-divergence---and proposes an estimation method grounded in these divergences. Section \ref{Properties of the Robust Estimators} examines the statistical properties of the estimators proposed in Section \ref{robust estimation}. Specifically, we demonstrate that the proposed estimators are consistent and asymptotically normal under appropriate regularity conditions. Additionally, we introduce the influence function to theoretically evaluate the estimators' robustness. Section \ref{Numerical Analysis} evaluates and validates the proposed methods through numerical experiments. This includes comparisons between the proposed and existing methods based on simulations under various scenarios and robustness assessment based on the influence function. The implementation R-code for the proposed methods is publicly available on GitHub (\url{https://github.com/yukiitaya/Robust_IRT.git}). Finally, in Section \ref{Conclusion and Future Directions}, we discuss the key findings and suggest potential avenues for future research.

\subsection{IRT Model}
Item response theory (IRT) posits that the probability of producing a correct (or otherwise specified) response to a test item is determined by the interplay between the respondent's latent trait and item-specific parameters (Baker \& Kim, \cite{baker2004}). Over the years, various IRT models have been developed to capture a broad array of item characteristics, spanning different response formats (e.g., dichotomous or polytomous) and levels of dimensionality (e.g., unidimensional or multidimensional). A core concept in IRT is the item characteristic curve (ICC), which represents the probability of a correct response to a particular item as a function of the respondent's latent trait (often termed ``ability''). This curve provides a visual and a quantitative means of understanding how item difficulty and respondent ability determine response probabilities.

In unidimensional IRT models for binary outcomes, the ICC is typically represented by a logistic function. One general form, known as the three-parameter logistic model (3PLM), can be expressed as
\begin{equation}
P_j(\theta; a_j,b_j,c_j) = c_j + (1-c_j)\frac{1}{1+\exp(-Da_j(\theta-b_j))} \ \ \ \ \ \ j= 1,\dots,J, \label{icc}
\end{equation}
where $P_j(\theta; a_j, b_j, c_j)$ denotes the conditional probability that an respondent given latent trait $\theta$ will respond correctly to item $j\in \{1,\dots, J\}$. The discrimination parameter $a_j$ controls how steeply the curve rises with respect to $\theta$, such that higher values of $a_j$ yield more pronounced changes in response probability with small shifts in $\theta$, thereby enhancing the item's ability to differentiate among respondents. The difficulty parameter $b_j$ indicates the level of $\theta$ at which the probability of a correct response is approximately 50\% (precisely 50\% when $c_j=0$). Higher values of $b_j$ correspond to more challenging items, requiring greater latent ability to achieve a moderate chance of success. The guessing parameter $c_j$ represents the probability of a correct response resulting solely from random guessing, even for individuals with very low ability. As $c_j$ increases, guessing is more influential in the likelihood of a correct response. Finally, $D$ is a scale factor, commonly set to $D = 1.702$, to align the logistic model's parameter scale with the normal ogive model.

IRT models are often categorized according to the number of parameters estimated. When all three parameters $(a_j,b_j,c_j)$ are freely estimated---including the guessing parameter $c_j$---the model is referred to as the three-parameter logistic model (3PLM). This specification provides considerable flexibility by accommodating item difficulty, discrimination, and the possibility of guessing. However, the 3PLM typically requires larger sample sizes for stable estimation and can be more challenging to interpret. In particular, its assumption that guessing probabilities remain constant across respondents (for each item) has been questioned in practice, leading to debates about the model's applicability. As a result, the 3PLM tends to be employed less frequently in practical testing contexts than models with fewer parameters (Embretson \& Reise, \cite{embretson2000}). In contrast, the two-parameter logistic model (2PLM) fixes the guessing parameter at zero ($c_j=0$) and estimates only the discrimination and difficulty parameters. By further constraining item discrimination to be identical across all items ($a_j=1, \, c_j=0$), one obtains the one-parameter logistic model (1PLM), also known as the Rasch model. Although more restrictive, the 1PLM provides significant simplicity in both interpretation and estimation, maintaining its foundational role in educational and psychological measurement.

In this study, we focus on estimation methods for the item parameters in the 1PLM. Suppose that a respondent with ability $\theta$ provides a binary response $u_j \in \{0,1\}$ to item $j \in \{1,\dots,J\}$.
Under the assumption that the responses follow the 1PLM, the probability function $q(u_j|\theta; b_j)$ is given by
\[
q(u_j|\theta; b_j) = P_j(\theta)^{u_j}Q_j(\theta)^{1-u_j} ,
\]
where $P_j(\theta)$ is the 1PLM's ICC with $a_j=1$ and $c_j=0$ in Eq.~\eqref{icc}:
\[
P_j(\theta) = P(u_j=1|\theta;b_j) = \frac{1}{1+\exp(-D(\theta - b_j))},
\]
and $Q_j(\theta)= 1-P_j(\theta)$ represents the probability of an incorrect response for an individual with ability $\theta$. 

In unidimensional IRT models,  it is assumed that one specific latent variable drives all item responses---a property known as unidimensionality, which can be tested via factor analyses or by examining residual structures (e.g., the $Q_3$ statistic proposed by Yen, \cite{yen1984}). If multiple, relatively independent factors are present, the fit of a unidimensional model is likely to suffer (Baker \& Kim, \cite{baker2004}).

Another crucial assumption in IRT is local independence, which stipulates that conditional on the same latent trait $\theta$, the responses across items are mutually independent. Under the assumption of local independence, the probability function for an individual's response pattern $\bm{u}=(u_1,\dots,u_J)$ given ability $\theta$ can be written as the product of the per-item response probabilities:
\begin{align}
q(\bm{u} | \theta; \bm{b}) = \prod_{j=1}^J q(u_j|\theta; b_j) = \prod_{j=1}^J P_j(\theta)^{u_j} Q_j(\theta)^{1-u_{j}}. \label{q}
\end{align}

In IRT models, item parameters are typically estimated using marginal maximum likelihood estimation (MMLE), which employs the expectation--maximization (EM) algorithm (Dempster et al., \cite{dempster1977}) to minimize Kullback--Leibler (KL) divergence. However, conventional MMLE is highly sensitive to aberrant responses, such as careless answering or random guessing (van der Linden \& Hambleton, \cite{vanderlinden1997}; Hong \& Cheng, \cite{hong2019}). This sensitivity can lead to significant reductions in estimation accuracy and reliability.

\section{Robust Estimation}\label{robust estimation}
To address this issue inherent in MMLE, this study introduces robust IRT parameter estimation methods. These approaches replace KL divergence with alternative divergences designed to enhance robustness, mitigating the impact of aberrant responses while maintaining high estimation accuracy. Focusing on the one-parameter logistic model, we formulate the objective function based on density power divergence and $\gamma$-divergence and optimize it via the majorization--minimization (MM) algorithm (a generalization of the EM algorithm) introduced by Lange et al. (\cite{lange2000}). In the MM algorithm, instead of directly optimizing the target objective function, a more tractable surrogate function (majorizer) is constructed. The algorithm ensures monotonic convergence toward the optimum by iteratively minimizing this surrogate function. 

\subsection{Robust Divergences}\label{robust divergence}

In statistical inference, robust divergences are crucial in analyzing data when model assumptions are violated or outliers are present. Among the various divergences proposed to date, density power divergence (DPD) and $\gamma$-divergence have garnered attention for their robustness and theoretical properties. These divergences generalize KL divergence while enhancing resistance to outliers and model misspecification.

DPD is defined as follows (Basu et al., \cite{basu1998}):
\begin{align}
D_{\beta}(p\| q) \coloneqq \frac{1}{\beta(1+\beta)} \int p(x)^{1+\beta} dx - \frac{1}{\beta} \int p(x)q(x)^{\beta} dx + \frac{1}{1+\beta} \int q(x)^{1+\beta} dx, \quad (\beta > 0). \label{DPD}
\end{align}
This divergence converges to KL divergence as $\beta \to 0$. With larger values of $\beta$, the influence of outliers diminishes, underscoring DPD's utility in linking KL divergence with robust estimation methodologies.

On the other hand, $\gamma$-divergence is defined as follows (Fujisawa \& Eguchi, \cite{fujisawa2008}):
\begin{align}
D_{\gamma}(p\| q) \coloneqq \frac{1}{\gamma(1+\gamma)} \log \int p(x)^{1+\gamma} dx - \frac{1}{\gamma} \log \int p(x)q(x)^{\gamma} dx + \frac{1}{1+\gamma} \log \int q(x)^{1+\gamma} dx, \quad (\gamma > 0). \label{gamma}
\end{align}
Similar to DPD, $\gamma$-divergence converges to KL divergence as $\gamma \to 0$. However, its incorporation of logarithmic transformations provides greater robustness against extreme outliers compared to DPD.

A common feature of both divergences is that the choice of the hyperparameters $\beta$ and $\gamma$ directly affects the trade-off between robustness and efficiency. Smaller values bring the divergence closer to KL divergence, improving estimation efficiency but reducing robustness. Conversely, larger values increase robustness to outliers at the cost of lowered efficiency.

\subsection{Divergence-Based Robust Estimation}\label{Divergence-Based Robust Estimation}

This subsection introduces robust estimation procedures based on DPD and $\gamma$-divergence, offering flexible robustness control via hyperparameter tuning. We begin with the DPD-based estimation procedure. Let $\bm{u}_1, \dots, \bm{u}_I$ be independent and identically distributed (i.i.d.) $J$-dimensional random vectors following the true distribution $p$, where each $\bm{u}_i$ represents the response pattern of respondent $i$ to $J$ items.
We examine the DPD-based discrepancy between $p$ and the parametric distribution $q(\bm{u}; \bm{b})$. Using the i.i.d. sample $\{\bm{u}_1, \dots, \bm{u}_I\}$, the parametric terms in $D_{\beta}(p\| q)$ \eqref{DPD} are empirically approximated by
\[
\ell_{\beta} (\bm{b}) \coloneqq -\frac{1}{\beta} \frac{1}{I} \sum_{i=1}^I q(\bm{u}_i ; \bm{b})^{\beta} +  \frac{1}{1+\beta} \int q(\bm{u} ; \bm{b})^{1+\beta} d\bm{u}.
\]
Here, the second term of $\ell_{\beta} (\bm{b})$ cannot be directly computed because the true distribution of $\bm{u}$ is unknown.
To address this, we introduce the following objective function $\tilde{\ell}_{\beta} (\bm{b})$ instead of $\ell_{\beta} (\bm{b})$:
\[
\tilde{\ell}_{\beta} (\bm{b}) \coloneqq  -\frac{1}{\beta} \frac{1}{I} \sum_{i=1}^I q(\bm{u}_i ; \bm{b})^{\beta} +  \frac{1}{1+\beta} \int q(\bm{u}|\theta ; \bm{b})^{1+\beta} f(\theta) d\theta d\bm{u},
\]
where $f(\theta)$ is the prior distribution of $\theta$, assumed to be standard normal. Based on the assumption of local independence, the second term of $\tilde{\ell}_{\beta} (\bm{b})$ can be expressed as
\begin{align*}
\int q(\bm{u}|\theta;\bm{b})^{1+\beta}f(\theta) d\theta d\bm{u} 
&= \int \bigg( \int \prod_{j=1}^J q(u_j|\theta;b_j)^{1+\beta} d\bm{u} \bigg) f(\theta) d\theta \\
&= \int \bigg(\prod_{j=1}^J \int  q(u_j|\theta;b_j)^{1+\beta} du_j \bigg) f(\theta) d\theta \\
&= \int \prod_{j=1}^J \left( q(1|\theta;b_j)^{1+\beta} + q(0|\theta;b_j)^{1+\beta} \right) f(\theta) d\theta.
\end{align*}
Since the final expression in the above equation does not depend on the unknown true distribution of $\bm{u}$, its computation is straightforward.
Given that the function $x \mapsto x^{1+\beta}$ is convex for $x \geq 0$, Jensen's inequality implies
\[
\tilde{\ell}_{\beta} (\bm{b}) \geq \ell_{\beta} (\bm{b}),
\]
with equality attained in the limit as $\beta \to 0$. 
Consequently, minimizing $\tilde{\ell}_{\beta} (\bm{b})$ results in a certain reduction in $\ell_{\beta} (\bm{b})$, and as $\beta \to 0$, the two minimization problems converge to equivalence. In the remainder of this subsection, we aim to minimize $\tilde{\ell}_{\beta} (\bm{b})$ instead of $\ell_{\beta} (\bm{b})$.

Next, we define the functional $\mathcal{L}_{\beta} : \mathcal{H}^I \times \mathbb{R}^J \to \mathbb{R}$ as
\[
\mathcal{L}_{\beta}(\bm{h}, \bm{b}) \coloneqq -\frac{1}{\beta}\frac{1}{I} \sum_{i=1}^I \int q(\bm{u}_i|\theta_i;\bm{b})^{\beta} h_i(\theta_i) d\theta_i + \frac{1}{1+\beta} \int q(\bm{u}|\theta ; \bm{b})^{1+\beta} f(\theta) d\theta d\bm{u}.
\]
The following discussion demonstrates that this functional serves as a majorizer for $\tilde{\ell}_{\beta} (\bm{b})$.
By considering the posterior distribution of $\theta_i$ given $\bm{u}_i$, denoted by $g(\theta_i|\bm{u}_i; \bm{b})$, and applying Bayes' theorem, we obtain
\[
q(\bm{u}_i|\theta_i; \bm{b}) = \frac{g(\theta_i|\bm{u}_i; \bm{b})q(\bm{u}_i; \bm{b})}{h_i(\theta_i)}.
\]
For a fixed $\beta$ with $0< \beta \leq 1$, the function $x \mapsto x^{\beta}$ is concave for $x \geq 0$. Thus, applying Jensen's inequality yields
\begin{align}
\mathcal{L}_{\beta}(\bm{h}, \bm{b}) - \tilde{\ell}_{\beta} (\bm{b}) &= -\frac{1}{\beta}\frac{1}{I} \sum_{i=1}^I \int q(\bm{u}_i|\theta_i;\bm{b})^{\beta} h_i(\theta_i) d\theta_i + \frac{1}{\beta} \frac{1}{I} \sum_{i=1}^I q(\bm{u}_i ; \bm{b})^{\beta} \notag \\
&= -\frac{1}{\beta}\frac{1}{I} \sum_{i=1}^I \int \left[ \frac{g(\theta_i|\bm{u}_i;\bm{b}) q(\bm{u}_i;\bm{b})}{h_i(\theta_i)} \right]^{\beta} h_i(\theta_i) d\theta_i + \frac{1}{\beta} \frac{1}{I} \sum_{i=1}^I q(\bm{u}_i ; \bm{b})^{\beta} \notag \\
&= \frac{1}{\beta}\frac{1}{I} \sum_{i=1}^I q(\bm{u}_i; \bm{b})^{\beta} \left( 1- \int \left[ \frac{g(\theta_i|\bm{u}_i;\bm{b})}{h_i(\theta_i)} \right]^{\beta} h_i(\theta_i) d\theta_i\right) \notag \\
&\geq \frac{1}{\beta}\frac{1}{I} \sum_{i=1}^I q(\bm{u}_i; \bm{b})^{\beta} \left( 1- \left[ \int \frac{g(\theta_i|\bm{u}_i;\bm{b})}{h_i(\theta_i)}  h_i(\theta_i) d\theta_i \right]^{\beta}\right) = 0. \notag
\end{align}
Thus, for any $\bm{h} \in \mathcal{H}^I$ and $\bm{b} \in \mathbb{R}^J$, we have
\begin{align}
\tilde{\ell}_{\beta} (\bm{b}) \leq \mathcal{L}_{\beta}(\bm{h}, \bm{b}). \label{aac}
\end{align}
The sufficient condition for equality is $\bm{h} = \bm{g}_{\bm{b}} \coloneqq (g(\cdot |\bm{u}_1;\bm{b}), \dots, g(\cdot |\bm{u}_I;\bm{b}))$.

The MM algorithm based on DPD can be constructed based on this property. First, set the initial value of the parameter $\bm{b}$ to $\bm{b}^{(0)}$, and for $t = 0, 1, \dots$, repeat the following two steps alternately until a convergence criterion is satisfied:
\begin{align}
\begin{aligned}
\text{majorization step:} &\quad  \bm{h}^{(t)} \in \argmin_{\bm{h} \in \mathcal{H}^I} \mathcal{L}_{\beta}(\bm{h},\bm{b}^{(t)}), \\
\text{minimization step:} & \quad  \bm{b}^{(t+1)} \in \argmin_{\bm{b}\in \mathbb{R}^J} \mathcal{L}_{\beta}(\bm{h}^{(t)},\bm{b}).
\end{aligned} \label{em_beta1}
\end{align}
Alternating between the majorization and minimization steps yields the sequence $\bm{h}^{(0)}, \bm{b}^{(1)}, \bm{h}^{(1)}, \dots$.

The equality condition in Eq.~\eqref{aac} implies that $\bm{h}^{(t)} = \bm{g}_{\bm{b}^{(t)}}$, leading to
\begin{align*}
& \mathcal{L}_{\beta}(\bm{h}^{(t)}, \bm{b}^{(t)}) = \tilde{\ell}_{\beta}(\bm{b}^{(t)}) \\
&= -\frac{1}{\beta} \frac{1}{I} \sum_{i=1}^I \int q(\bm{u}_i | \theta_i ; \bm{b}^{(t)})^{\beta} g(\theta_i|\bm{u}_i;\bm{b}^{(t)}) d\theta_i +  \frac{1}{1+\beta} \int q(\bm{u}|\theta ; \bm{b}^{(t)})^{1+\beta} f(\theta) d\theta d\bm{u}.
\end{align*}
Furthermore, the sequence $\{\bm{b}^{(t)}\}$ generated by this algorithm provides a monotonically non-increasing sequence of the objective function $\tilde{\ell}_{\beta}(\bm{b})$:
\begin{align}
\tilde{\ell}_{\beta}(\bm{b}^{(0)}) \geq \cdots \geq \tilde{\ell}_{\beta}(\bm{b}^{(t)}) \geq \tilde{\ell}_{\beta}(\bm{b}^{(t+1)}) \geq \cdots. \label{seq_beta}
\end{align}
Since $\tilde{\ell}_{\beta}(\bm{b})$ is bounded below, the sequence \eqref{seq_beta} converges to some value $\tilde{\ell}_{\beta}(\bm{\hat{b}}_{\beta})$. 
Additionally, since we have
\[
\bm{\hat{b}}_{\beta} \in \argmin_{\bm{b} \in \mathbb{R}^J} \mathcal{L}_{\beta} (\bm{g}_{\bm{b}}, \bm{\hat{b}}_{\beta}), \quad \bm{\hat{b}}_{\beta} \in \argmin_{\bm{b} \in \mathbb{R}^J} \mathcal{L}_{\beta} (\bm{g}_{\bm{\hat{b}}_{\beta}}, \bm{b}),
\]
then it follows 
\[
\left. \frac{\partial}{\partial \bm{b}} \tilde{\ell}_{\beta}(\bm{b}) \right|_{\bm{b}=\bm{\hat{b}}_{\beta}} = \left. \frac{\partial}{\partial \bm{b}} \mathcal{L}_{\beta} (\bm{g}_{\bm{b}}, \bm{\hat{b}}_{\beta}) \right|_{\bm{b}=\bm{\hat{b}}_{\beta}} + \left. \frac{\partial}{\partial \bm{b}} \mathcal{L}_{\beta} (\bm{g}_{\bm{\hat{b}}_{\beta}}, \bm{b})  \right|_{\bm{b}=\bm{\hat{b}}_{\beta}} = \bm{0},
\]
indicating that $\bm{\hat{b}}_{\beta}$ is a stationary point of $\tilde{\ell}_{\beta}(\bm{b})$.

Thus, by appropriately selecting the initial value $\bm{b}^{(0)}$, the algorithm \eqref{em_beta1} is guaranteed to converge to an estimator that sufficiently minimizes the empirical DPD. The pseudocode in Algorithm \ref{algo2} provides a formal reformulation of the MM algorithm \eqref{em_beta1}.
\begin{algorithm}
\caption{MM Algorithm Based on Density Power Divergence}
\label{algo2}
\begin{algorithmic}[1]
\State \textbf{Input: } i.i.d. sample $\{\bm{u}_1, \dots, \bm{u}_I\}$,\hspace{2mm} $\bm{b}^{(0)} \in \mathbb{R}^J$,\hspace{2mm} $0<\beta \leq 1$
\State \textbf{Initialize: } Set the initial value $\bm{b}^{(0)}$.
\Repeat
    \State \textbf{majorization step:}
    \Statex \hspace{\algorithmicindent}
    \[
    \bm{h}^{(t)} = \bm{g}_{\bm{b}^{(t)}} \coloneqq 
    \big( g(\cdot | \bm{u}_1; \bm{b}^{(t)}), \dots, g(\cdot | \bm{u}_I; \bm{b}^{(t)}) \big) \in \argmin_{\bm{h}\in \mathcal{H}^I} \mathcal{L}_{\beta}(\bm{h}, \bm{b}^{(t)}).
    \]
    \State \textbf{minimization step:}
    \Statex \hspace{\algorithmicindent} 
    \[
    \bm{b}^{(t+1)} \in \argmin_{\bm{b}\in \mathbb{R}^J} \mathcal{L}_{\beta}(\bm{g}_{\bm{b}^{(t)}}, \bm{b}).
    \]
    \State Update $t \gets t + 1$.
\Until{ The convergence criterion is satisfied.}
\State \textbf{Output: } Output the estimator $\bm{\hat{b}}_{\beta} = \bm{b}^{(t)}$.
\end{algorithmic}
\end{algorithm}

Similar to the case of DPD, the same discussion applies to $\gamma$-divergence. In the case of $\gamma$-divergence, the empirical approximation of the terms involving $q$ is given by
\[
d_{\gamma} (p \| q) \simeq -\frac{1}{\gamma} \log \left( \frac{1}{I} \sum_{i=1}^I q(\bm{u}_i ; \bm{b})^{\gamma}\right) +  \frac{1}{1+\gamma} \log \int q(\bm{u} ; \bm{b})^{1+\gamma} d\bm{u} \eqqcolon \ell_{\gamma} (\bm{b}).
\]
We introduce the following objective function $\tilde{\ell}_{\gamma}(\bm{b})$ as
\[
\tilde{\ell}_{\gamma}(\bm{b}) \coloneqq -\frac{1}{\gamma} \log \left( \frac{1}{I} \sum_{i=1}^I q(\bm{u}_i ; \bm{b})^{\gamma}\right) +  \frac{1}{1+\gamma} \log \int q(\bm{u}|\theta ; \bm{b})^{1+\gamma} f(\theta) d\theta d\bm{u} .
\]
As in the case of DPD, the objective function $\tilde{\ell}_{\gamma}(\bm{b})$ is an upper bound of the function $\ell_{\gamma}(\bm{b})$, derived from Jensen's inequality. Furthermore, we define the functional $\mathcal{L}_{\gamma} : \mathcal{H}^I \times \mathbb{R}^J \to \mathbb{R}$ as
\[
\mathcal{L}_{\gamma}(\bm{h}, \bm{b}) \coloneqq -\frac{1}{\gamma} \log \bigg( \frac{1}{I} \sum_{i=1}^I \int q(\bm{u}_i|\theta_i;\bm{b})^{\gamma} h_i(\theta_i) d\theta_i \bigg) + \frac{1}{1+\gamma} \log \int q(\bm{u}|\theta ; \bm{b})^{1+\gamma} f(\theta) d\theta d\bm{u}.
\]
Applying Jensen's inequality, it follows that for $0<\gamma \leq 1$, the functional $\mathcal{L}_{\gamma}(\bm{h}, \bm{b})$ serves as a majorizer for $\tilde{\ell}_{\gamma}(\bm{b})$. Analogous to the DPD case, an MM algorithm can be formulated, as depicted in Algorithm \ref{algo3}. The resulting estimator $\bm{\hat{b}}_{\gamma}$ is a stationary point of the objective function $\tilde{\ell}_{\gamma}(\bm{b})$.

\begin{algorithm}
\caption{MM Algorithm Based on $\gamma$-Divergence}
\label{algo3}
\begin{algorithmic}[1]
\State \textbf{Input: } i.i.d. sample $\{\bm{u}_1, \dots, \bm{u}_I\}$,\hspace{2mm} $\bm{b}^{(0)} \in \mathbb{R}^J$, \hspace{2mm} $0<\gamma \leq 1$
\State \textbf{Initialize: } Set the initial value $\bm{b}^{(0)}$.
\Repeat
    \State \textbf{majorization step:}
    \Statex \hspace{\algorithmicindent} 
    \[
    \bm{h}^{(t)} = \bm{g}_{\bm{b}^{(t)}} \coloneqq 
    \big( g(\cdot | \bm{u}_1; \bm{b}^{(t)}), \dots, g(\cdot | \bm{u}_I; \bm{b}^{(t)}) \big) \in \argmin_{\bm{h}\in \mathcal{H}^I} \mathcal{L}_{\gamma}(\bm{h}, \bm{b}^{(t)}).
    \]
    \State \textbf{minimization step:}
    \Statex \hspace{\algorithmicindent} 
    \[
    \bm{b}^{(t+1)} \in \argmin_{\bm{b}\in \mathbb{R}^J} \mathcal{L}_{\gamma}(\bm{g}_{\bm{b}^{(t)}}, \bm{b}).
    \]
    \State Update $t \gets t + 1$.
\Until{ The convergence criterion is satisfied.}
\State \textbf{Output: } Output the estimator $\bm{\hat{b}}_{\gamma} = \bm{b}^{(t)}$.
\end{algorithmic}
\end{algorithm}

\section{Properties of the Robust Estimators}\label{Properties of the Robust Estimators}

This section examines the statistical properties of the robust estimators proposed in Section \ref{Divergence-Based Robust Estimation}. First, by formulating these estimators as $M$-estimators, we establish their consistency and asymptotic normality under appropriate regularity conditions. Additionally, we theoretically show that as the hyperparameters approach zero, the estimators based on DPD and $\gamma$-divergence converge to conventional MMLE, which is based on KL divergence. These analyses provide essential insights into the trade-off between robustness and efficiency. Next, we introduce the influence function to theoretically assess the estimators' robustness against outliers. Throughout this section, $\alpha$ denotes either $\beta$ or $\gamma$, and we discuss the estimators based on DPD and $\gamma$-divergence in parallel.

\subsection{Asymptotic Properties}\label{Asymptotic Properties}

First, let $\{\bm{u}_1, \dots, \bm{u}_I\}$ be independent samples drawn from the true distribution $p$. Note that $p$ is not necessarily included in the parametric model $\{ \bm{u} \mapsto q(\bm{u}; \bm{b}) \mid \bm{b} \in \mathbb{R}^J \}$. The estimators $\bm{\hat{b}}_{\beta}$ and $\bm{\hat{b}}_{\gamma}$, introduced in Section \ref{Divergence-Based Robust Estimation}, are stationary points of the following functions, respectively:
\begin{align*}
\begin{aligned}
M^{(I)}_{\beta}(\bm{b}) &\coloneqq \mathcal{L}_{\beta}(\bm{g}_{\bm{\hat{b}}_{\beta}}, \bm{b}) = -\frac{1}{\beta} \frac{1}{I} \sum_{i=1}^I \int q(\bm{u}_i|\theta; \bm{b})^{\beta} g(\bm{u}_i|\theta; \bm{\hat{b}}_{\beta}) d\theta + \frac{1}{1+\beta} \int q(\bm{u}|\theta; \bm{b})^{1+\beta} f(\theta) d\theta d\bm{u}, \\
M^{(I)}_{\gamma}(\bm{b}) &\coloneqq \mathcal{L}_{\gamma}(\bm{g}_{\bm{\hat{b}}_{\gamma}}, \bm{b}) = -\frac{1}{\gamma} \log \left( \frac{1}{I} \sum_{i=1}^I \int q(\bm{u}_i|\theta; \bm{b})^{\gamma} g(\bm{u}_i|\theta; \bm{\hat{b}}_{\gamma}) d\theta \right) + \frac{1}{1+\gamma} \log \int q(\bm{u}|\theta; \bm{b})^{1+\gamma} f(\theta) d\theta d\bm{u}.
\end{aligned} 
\end{align*}
Next, we define the following vector-valued functions as
\setlength{\jot}{5pt}
\begin{align}
&\begin{aligned}
\psi_{\beta}(\bm{b}; \bm{u}) \coloneqq 
& - \frac{\int q(\bm{u}|\theta;\bm{b})^{1+\beta} \xi(\bm{u},\theta,\bm{b}) f(\theta)d\theta}{q(\bm{u};\bm{b}) } + \int  q(\bm{u}|\theta;\bm{b})^{1+\beta} \xi(\bm{u},\theta,\bm{b}) f(\theta) d\theta d\bm{u}, 
\end{aligned} \label{psi_beta}  \\
&\begin{aligned}
\psi_{\gamma}(\bm{b}; \bm{u}) \coloneqq & - \frac{\int q(\bm{u}|\theta;\bm{b})^{1+\gamma} \xi(\bm{u},\theta,\bm{b}) f(\theta)d\theta}{q(\bm{u};\bm{b})} \int q(\bm{u}|\theta;\bm{b})^{1+\gamma} f(\theta) d\theta d\bm{u} \\ 
& + \frac{\int q(\bm{u}|\theta;\bm{b})^{1+\gamma} f(\theta)d\theta}{q(\bm{u};\bm{b})} \int q(\bm{u}|\theta;\bm{b})^{1+\gamma} \xi(\bm{u},\theta,\bm{b}) f(\theta) d\theta d\bm{u},
\end{aligned} \label{psi_gamma} 
\end{align}
\setlength{\jot}{3pt}
where, $\xi(\bm{u}, \theta, \bm{b})$ is a $J$-dimensional vector-valued function whose $j$th component is given by
\[
D(P_j(\theta)-1)^{u_j} P_j(\theta)^{1-u_j}.
\]
For $\alpha \in \{\beta, \gamma\}$, the following equivalence holds:
\[
\nabla M^{(I)}_{\alpha}(\bm{b}) = \bm{0} \quad \Leftrightarrow \quad \Psi^{(I)}_{\alpha}(\bm{b}) \coloneqq \frac{1}{I} \sum_{i=1}^I \psi_{\alpha}(\bm{b}; \bm{u}_i) = \bm{0}.
\]
Since $\bm{\hat{b}}_{\alpha}$ is a stationary point of $M^{(I)}_{\alpha}(\bm{b})$, it follows that
\begin{align}
\Psi^{(I)}_{\alpha}(\bm{\hat{b}}_{\alpha}) = \bm{0}. \label{est_eq}
\end{align}
Let $\bm{b}^*_{\alpha}$ be the solution satisfying the following expectation-based equation:
\begin{align}
\Psi_{\alpha}(\bm{b}) \coloneqq \E [\psi_{\alpha}(\bm{b}; \bm{u})] = \bm{0}. \label{id_eq}
\end{align}
The estimator $\bm{\hat{b}}_{\alpha}$ is an $M$-estimator of $\bm{b}^*_{\alpha}$. Under appropriate regularity conditions, it can be shown that $\bm{\hat{b}}_{\alpha}$ is consistent and asymptotically normal. The following theorem is based on the discussion in van der Vaart (\cite{vandervaart2000}).
\begin{theorem}\label{thm}
    Suppose there exists a compact set $\mathcal{B}$ such that $\bm{b}^*_{\alpha} \in \mathcal{B}^{\circ}$ is the unique solution to the equation 
$\Psi_{\alpha}(\bm{b}) = \bm{0}$, where $\mathcal{B}^{\circ}$ denotes the interior of $\mathcal{B}$. Then, for any $\bm{\hat{b}}_{\alpha}$ satisfying Eq.~\eqref{est_eq}, the following holds:
\[
\bm{\hat{b}}_{\alpha} \xrightarrow{P} \bm{b}^*_{\alpha}, \quad \text{as }  \ \ I \to \infty,
\]
where $\xrightarrow{P}$ denotes convergence in probability. Furthermore, if the following matrix-valued function
\[
V_{\alpha}(\bm{b}) \coloneqq \E \biggl[ \frac{\partial \psi_{\alpha}(\bm{b}; \bm{u})}{\partial \bm{b}^{\top}} \biggr]
\]
is nonsingular at $\bm{b}^*_{\alpha}$, then $\bm{\hat{b}}_{\alpha}$ is asymptotically normal:
\begin{align}
\sqrt{I} \bigl( \bm{\hat{b}}_{\alpha} - \bm{b}^*_{\alpha} \bigr) \rightsquigarrow \mathcal{N} \Bigl( \bm{0},\, V_{\alpha}(\bm{b}^*_{\alpha})^{-1} K_{\alpha}(\bm{b}^*_{\alpha}) \bigl(V_{\alpha}(\bm{b}^*_{\alpha})^{-1} \bigr)^{\top} \Bigr), \quad \text{as } \ \ I \to \infty, \label{normality}
\end{align}
where $\rightsquigarrow$ denotes convergence in distribution and 
\[
K_{\alpha}(\bm{b}) \coloneqq \E \bigl[ \psi_{\alpha}(\bm{b}; \bm{u}) \psi_{\alpha}(\bm{b}; \bm{u})^{\top} \bigr].
\]
\end{theorem}

\begin{proof}
First, by the Glivenko-Cantelli theorem (as presented in Chapter 19 of van der Vaart, \cite{vandervaart2000}), the following uniform convergence holds:
\[
\sup_{\bm{b} \in \mathcal{B}} \big\| \Psi^{(I)}_{\alpha}(\bm{b}) - \Psi_{\alpha}(\bm{b}) \big\| \xrightarrow{P} 0, \quad \text{as} \ \ I \to \infty,
\]
where $\|\cdot\|$ denotes the $L^2$-norm. Additionally, by the uniqueness of $\bm{b}^*_{\alpha}$ and the compactness of $\mathcal{B}$, for any $\epsilon > 0$, the following holds:
\[
\inf_{\bm{b} \in \mathcal{B} : \|\bm{b} - \bm{b}^*_{\alpha}\| \geq \epsilon} \|\Psi_{\alpha}(\bm{b})\| > 0 = \|\Psi_{\alpha}(\bm{b}^*_{\alpha})\|.
\]
Therefore, by Theorem 5.7 of van der Vaart (\cite{vandervaart2000}), we obtain
\[
\bm{\hat{b}}_{\alpha} \xrightarrow{P} \bm{b}^*_{\alpha}, \quad \text{as} \ \ I \to \infty.
\]

Next, to establish the asymptotic normality of the estimator $\bm{\hat{b}}_{\alpha}$, we verify the assumptions of Theorem 5.21 in van der Vaart (\cite{vandervaart2000}). First, by the multivariate Taylor theorem, there exists $\eta \in [0, 1]$, and $\bm{\tilde{b}} = \eta \bm{b}_1 + (1-\eta) \bm{b}_2$, such that
\begin{align}
\psi_{\alpha} (\bm{b}_1; \bm{u}) - \psi_{\alpha} (\bm{b}_2; \bm{u}) = V_{\alpha}(\bm{\tilde{b}}; \bm{u})(\bm{b}_1 - \bm{b}_2), \label{taylor}
\end{align}
where $V_{\alpha}(\bm{b}; \bm{u})$ is the Jacobian matrix of the function $\bm{b} \mapsto \psi_{\alpha}(\bm{b}; \bm{u})$:
\begin{equation}
V_{\alpha}(\bm{b}; \bm{u}) \coloneqq \frac{\partial \psi_{\alpha}(\bm{b}; \bm{u})}{\partial \bm{b}^{\top}}. \label{V_alpha}
\end{equation}
By the continuity of $V_{\alpha}(\bm{b}; \bm{u})$ with respect to $\bm{b}$ and the compactness of $\mathcal{B}$, we obtain
\begin{align}
\bar{V}_{\alpha} \coloneqq \sup_{\bm{b} \in \mathcal{B}, \bm{u} \in \{0,1\}^J} \|V_{\alpha}(\bm{b}; \bm{u}) \|_{\mathrm{F}} < \infty, \label{J}
\end{align}
where $\|\cdot\|_{\mathrm{F}}$ is the Frobenius norm. Using Eq.~\eqref{taylor} and \eqref{J}, the following inequality holds for any $\bm{u} \in \{0,1\}^J$:
\[
\big\| \psi_{\alpha} (\bm{b}_1; \bm{u}) - \psi_{\alpha} (\bm{b}_2; \bm{u}) \big\| \leq \bar{V}_{\alpha} \|\bm{b}_1 - \bm{b}_2\|.
\]
Furthermore, since $\psi_{\alpha} (\bm{b}^*_{\alpha}; \bm{u})$ is finite for any $\bm{u} \in \{0,1\}^J$, the following holds:
\[
\E \big[ \|\psi_{\alpha} (\bm{b}^*_{\alpha}; \bm{u}) \|^2 \big] < \infty.
\]
Combining these facts, the assumptions of Theorem 5.21 in van der Vaart (\cite{vandervaart2000}) are satisfied, establishing the asymptotic normality of the estimator $\bm{\hat{b}}_{\alpha}$. That is,
\[
\sqrt{I} \bigl( \bm{\hat{b}}_{\alpha} - \bm{b}^*_{\alpha} \bigr) \rightsquigarrow \mathcal{N} \Bigl(\bm{0},\, V_{\alpha}(\bm{b}^*_{\alpha})^{-1} K_{\alpha}(\bm{b}^*_{\alpha}) \big(V_{\alpha}(\bm{b}^*_{\alpha})^{-1} \big)^{\top} \Bigr), \quad \text{as } \ \ I \to \infty.
\]
\end{proof}

By taking the limit as $\alpha \to 0$, the above discussion converges to the framework of marginal maximum likelihood estimation (MMLE). By reasoning in the same way as in Section \ref{Divergence-Based Robust Estimation}, it follows that the marginal maximum likelihood estimator $\bm{\hat{b}}_{\text{KL}}$ is a stationary point of the function:
\[
M_{\text{KL}}^{(I)}(\bm{b}) \coloneqq -\frac{1}{I} \sum_{i=1}^I \int \log q(\bm{u}_i | \theta_i; \bm{b}) g(\theta |\bm{u}_i; \bm{\hat{b}}_{\text{KL}}) \, d\theta,
\]
Therefore, $\bm{\hat{b}}_{\text{KL}}$ satisfy the following estimating equation:
\[
\frac{1}{I} \sum_{i=1}^I \psi_{\text{KL}}(\bm{\hat{b}}_{\text{KL}}; \bm{u}_i) = \bm{0},
\]
where 
\[
\psi_{\text{KL}}(\bm{b}; \bm{u}) \coloneqq -\frac{\int q(\bm{u} | \theta; \bm{b}) \xi(\bm{u}, \theta, \bm{b}) f(\theta) d\theta}{q(\bm{u}; \bm{b})}.
\]
We can verify that, for $\alpha \in \{\beta, \gamma\}$, the following limits hold:
\[
\lim_{\alpha \to 0} \psi_{\alpha}(\bm{b}; \bm{u}) = \psi_{\text{KL}}(\bm{b}; \bm{u}) \quad \text{and} \quad
\lim_{\alpha \to 0} V_{\alpha}(\bm{b}; \bm{u}) = V_{\text{KL}}(\bm{b}; \bm{u}),
\]
where $V_{\alpha}(\bm{b}; \bm{u})$ is defined in Eq.~\ref{V_alpha}, and $V_{\text{KL}}(\bm{b}; \bm{u})$ is the Jacobian matrix of the function $\bm{b} \mapsto \psi_{\text{KL}}(\bm{b}; \bm{u})$. The explicit expressions for $V_{\alpha}(\bm{b}; \bm{u})$ and $V_{\text{KL}}(\bm{b}; \bm{u})$ are presented in Appendix \ref{appendixA}. It can be confirmed that the asymptotic variance of the estimator $\hat{\bm{b}}_{\alpha}$ obtained in Theorem \ref{thm} as $\alpha\to 0$ coincides with the sandwich-type covariance matrix for the marginal maximum likelihood estimator derived by Yuan et al.~(\cite{yuan2014}).

Furthermore, the convergence of $\hat{\bm{b}}_{\alpha}$ to $\hat{\bm{b}}_{\text{KL}}$ as $\alpha \to 0$ under appropriate regularity conditions can be established as follows. For each $i \in \{1, \dots, I\}$, applying the multivariate Taylor theorem yields the existence of some $\eta_i \in [0,1]$ and an intermediate point $\bm{\tilde{b}}_i = \eta_i \hat{\bm{b}}_{\text{KL}} + (1-\eta_i) \hat{\bm{b}}_{\alpha}$ such that  
\[
\psi_{\alpha} (\hat{\bm{b}}_{\text{KL}}; \bm{u}_i) - \psi_{\alpha} (\hat{\bm{b}}_{\alpha}; \bm{u}_i) = V_{\alpha}(\bm{\tilde{b}}_i; \bm{u}_i)(\hat{\bm{b}}_{\text{KL}} - \hat{\bm{b}}_{\alpha}).
\]
Summing over $i=1,\dots,I$ leads to  
\[
\sum_{i=1}^I \psi_{\alpha} (\hat{\bm{b}}_{\text{KL}}; \bm{u}_i) = \sum_{i=1}^I V_{\alpha}(\bm{\tilde{b}}_i; \bm{u}_i)(\hat{\bm{b}}_{\text{KL}} - \hat{\bm{b}}_{\alpha}).
\]
Assuming that the matrix $\sum_{i=1}^I V_{\alpha}(\bm{\tilde{b}}_i; \bm{u}_i)$ is nonsingular for any sufficiently small $\alpha \geq 0$, it follows that  
\begin{align*}
\big\| \hat{\bm{b}}_{\text{KL}} - \hat{\bm{b}}_{\alpha} \big\| =\ & \bigg\| \sum_{i=1}^I V_{\alpha}(\bm{\tilde{b}}_i; \bm{u}_i) \bigg\|_{\mathrm{F}}^{-1} \bigg\| \sum_{i=1}^I \psi_{\alpha} (\hat{\bm{b}}_{\text{KL}}; \bm{u}_i) \bigg\| \\
\to \ & \bigg\| \sum_{i=1}^I V_{\text{KL}}(\bm{\tilde{b}}_i; \bm{u}_i) \bigg\|_{\mathrm{F}}^{-1} \bigg\| \sum_{i=1}^I \psi_{\text{KL}} (\hat{\bm{b}}_{\text{KL}}; \bm{u}_i) \bigg\| 
= 0, \quad \text{as} \quad \alpha \to 0
\end{align*}
This establishes the pointwise convergence:
\[
\lim_{\alpha \to 0} \hat{\bm{b}}_{\alpha} = \hat{\bm{b}}_{\text{KL}}.
\]
This connection highlights an essential property of the proposed robust estimation framework: as the hyperparameters approach zero, the estimators based on DPD or $\gamma$-divergence converge to the marginal maximum likelihood estimator. Consequently, the framework extends the conventional approach and facilitates a seamless transition between the proposed and existing estimation methods.

It is noteworthy that the $j$th element of $\psi_{\text{KL}}(\bm{b}; \bm{u})$ is given by
\begin{equation}
 \frac{\int q(u_j|\theta; b_j)(P_j(\theta)-1)^{1-u_j}P_j(\theta)^{1-u_j}f(\theta) d\theta}{\int q(u_j|\theta;b_j)f(\theta)d\theta}, \label{mmle_eq}
\end{equation}
which depends solely on $b_j$. This implies that, under MMLE, the difficulty parameter $b_j$ is exclusively estimated based on the responses to item $j$. In other words, the estimator for item $j$ contains only the ``information'' derived from responses to the specific item $j$. On the other hand, the estimators based on DPD or $\gamma$-divergence incorporate ``information'' regarding how the respondent has answered other items. This distinction contributes significantly to the robustness of the proposed estimation approaches, as they utilize the broader context of response patterns across all items rather than focusing solely on individual item responses. It should be noted that this property also leads to an increase in computational complexity.

\subsection{Influence Function}\label{Influence Function}
In this subsection, we introduce the influence function to theoretically evaluate the robustness of the estimators proposed in Section \ref{Divergence-Based Robust Estimation}. The influence function is an essential tool for assessing the robustness and sensitivity of statistics (Hampel et al., \cite{hampel1986}). In particular, it is used to quantitatively analyze the extent to which a statistic is affected by outliers or anomalous data. The influence function provides a differential representation of how a statistic $T(p)$ responds to changes in the underlying distribution $p$, offering a key theoretical foundation for evaluating the properties of statistics.

The influence function is generally defined as follows. For a population distribution $p$, the influence function at a point $x \in \mathcal{X}$ is defined as
\begin{align}
\mathrm{IF}(x; T, p) = \lim_{\epsilon \to 0} \frac{T((1 - \epsilon)p + \epsilon \Delta_x) - T(p)}{\epsilon}. \label{if}
\end{align}
Here, $\Delta_x$ denotes a probability measure with mass one at point $x$. This definition illustrates how the statistic $T$ changes when the data distribution is slightly perturbed.

Gross-error sensitivity stands out as the key criterion for evaluating the robustness of statistics using the influence function, as outlined in Hampel et al.~(\cite{hampel1986}). The gross-error sensitivity $\kappa^*(T;p)$ measures the maximum impact that a single observation can have on $T(p)$ and is defined as follows:
\[
\kappa^*(T; p) \coloneqq \sup_{x \in \mathcal{X}} \| \mathrm{IF}(x; T,p) \|.
\]
The gross-error sensitivity quantifies the worst-case (approximate) influence of fixed-sized small contamination in the data distribution on the value of the statistic $T(p)$. Thus, it can be interpreted as the upper bound of the (standardized) asymptotic bias of the statistic. The smaller the value of $\kappa^*$, the less sensitive the statistic is to outliers, allowing for more stable estimation.

When a statistic is defined as an $M$-estimator that solves an equation system like \eqref{est_eq},
\[
\frac{1}{I} \sum_{i=1}^I \psi_{\alpha}(\bm{b}; \bm{u}_i) = \bm{0},
\]
the influence function of the estimator $\bm{\hat{b}}_{\alpha}$ at a point $\bm{u} \in \{0,1\}^J$ for the population distribution $p$ can be expressed as follows (Hampel et al., \cite{hampel1986}):
\begin{align}
\mathrm{IF}(\bm{u}; \bm{\hat{b}}_{\alpha}, p) = -V_{\alpha}(\bm{\hat{b}}_{\alpha})^{-1} \psi_{\alpha}(\bm{\hat{b}}_{\alpha}, \bm{u}). \label{influence_function}
\end{align}
The influence functions of the estimators based on marginal maximum likelihood, DPD, and $\gamma$-divergence are numerically analyzed and discussed in Section \ref{Analysis Based on Influence Functions}.

\section{Numerical Analysis}\label{Numerical Analysis}
This section empirically evaluates and validates the robust estimation methods proposed in Section \ref{Divergence-Based Robust Estimation}. First, in Section \ref{Simulation}, we perform extensive simulation studies under diverse conditions, including scenarios featuring guessing responses that deviate from the 1PLM. These studies compare the proposed methods with conventional MMLE and robust MMLE introduced by Hong and Cheng (\cite{hong2019}). Robust MMLE is an estimation method designed to mitigate bias and downweight the influence of outliers by maximizing a weighted marginal likelihood, where the weights are derived from standardized individual fit statistics. In Section \ref{Analysis Based on Influence Functions}, we investigate the influence functions of individual observations, providing further insights into the practical performance of the proposed robust estimation methods. To ensure reproducibility, the implementation R-code for the proposed methods is available on GitHub (\url{https://github.com/yukiitaya/Robust_IRT.git}).

\subsection{Simulation}\label{Simulation}
\subsubsection{Setup}\label{Setup}
In this simulation, the responders' ability parameters $\theta_1,\dots,\theta_I$ are independently generated from the standard normal distribution. The item difficulty parameters are set by evenly dividing the interval $[-2, 2]$ into $J$ partitions, ensuring coverage of a broad range of ability levels. We consider two sample sizes, $I=500$ and $I=1000$, and two test lengths, $J=15$ and $J=30$. Given these conditions, the typical responses are generated according to 1PLM in Eq.~\eqref{q}.

Aberrant responses are introduced into the data by incorporating two types of guessing behavior: an Unbiased type (the probability of 0 or 1 is equally 0.5) and a Biased type (the probabilities of 0 and 1 are 0.8 and 0.2, respectively). The Unbiased type simulates completely random responses on binary-choice items, while the Biased type represents a scenario akin to random responses on five-choice items.

Regarding the occurrence of guessing responses, we examine three scenarios:
\begin{itemize}
\item \textbf{Scenario 1}: The probability of guessing is uniform with respect to $\theta$.
\item \textbf{Scenario 2}: The probability of guessing depends on $\theta$.
\item \textbf{Scenario 3}: No guessing takes place.
\end{itemize}

In Scenario 1, as outlined by Hong and Cheng (\cite{hong2019}), varying proportions of aberrant responses are simulated by adjusting two key parameters:
\begin{itemize}
\item Prevalence (Prev.): the proportion of individuals whose responses are affected by guessing.
\item Severity (Sev.): the proportion of items for which these individuals provide random responses.
\end{itemize}
Let $Z_{\text{Prev.}},\, Z_{\text{Sev.}}$ and $Z_r$ be independent random variables, each following a Bernoulli distribution with parameters Prev., Sev., and $r$, respectively. Then, for each respondent $i$, the observed response pattern $\bm{u}'_i$ is formally expressed as 
\[
\bm{u}'_i = (1-Z_{\text{Prev.}})\, \bm{u}_i + Z_{\text{Prev.}}\, \bm{\tilde{u}}_i,
\]
where $\bm{u}_i=(u_{i1},\dots,u_{iJ})$ represents a normal response pattern following 1PLM, and $\bm{\tilde{u}}_i = (\tilde{u}_{i1}, \dots, \tilde{u}_{iJ})$ denotes a response pattern contaminated by guessing:
\[
\tilde{u}_{ij} = (1-Z_{\text{Sev.}})\, u_{ij} + Z_{\text{Sev.}}\, Z_r, \quad j=1,\dots,J.
\]
Here, $Z_r$ represents the guessing response. For the Unbiased type, $r=0.5$, whereas for the Biased type, $r=0.2$. Once a responder is flagged as a guesser, the probability for each item is shifted from the 1PLM-based value $u_{ij}$ toward the constant guess parameter $Z_r$, with the severity of this shift controlled by $Z_{\text{Sev.}}$. In this simulation, we examine combinations of $\text{Prev.}\in \{0.1, 0.3\}$ and $\text{Sev.}\in \{0.3, 0.5\}$. For example, when $(\text{Prev.}, \text{Sev.})=(0.1, 0.3)$, randomly selected 10\% of responders guess on 30\% of their items.

In Scenario 2, the probability of guessing for a responder with ability $\theta$ is determined by one of the following two functions:
\[
R_1(\theta)=\frac{1}{1+\exp(0.5(\theta+5))}, \quad R_2(\theta)=\frac{1}{1+\exp(1.5(\theta+2))}.
\]
These functions are illustrated in Figure \ref{fig:r_function}. The function $R_1(\theta)$ captures a moderate dependency where guessing gradually increases across a wide ability range. By contrast, $R_2(\theta)$ represents a sharper increase in guessing for lower ability levels. 
In this scenario, the response pattern $\bm{u}'_i=(u'_{i1},\dots,u'_{iJ})$ for an individual with ability $\theta_i$ is
\[
u'_{ij} = (1-Z_{R_{l}(\theta_i)})\, u_{ij} + Z_{R_{l}(\theta_i)}\, Z_r, \quad j=1,\dots,J, \quad l=1,2.
\]
\begin{figure}[htbp]
\centering
\includegraphics[width=0.8\textwidth]{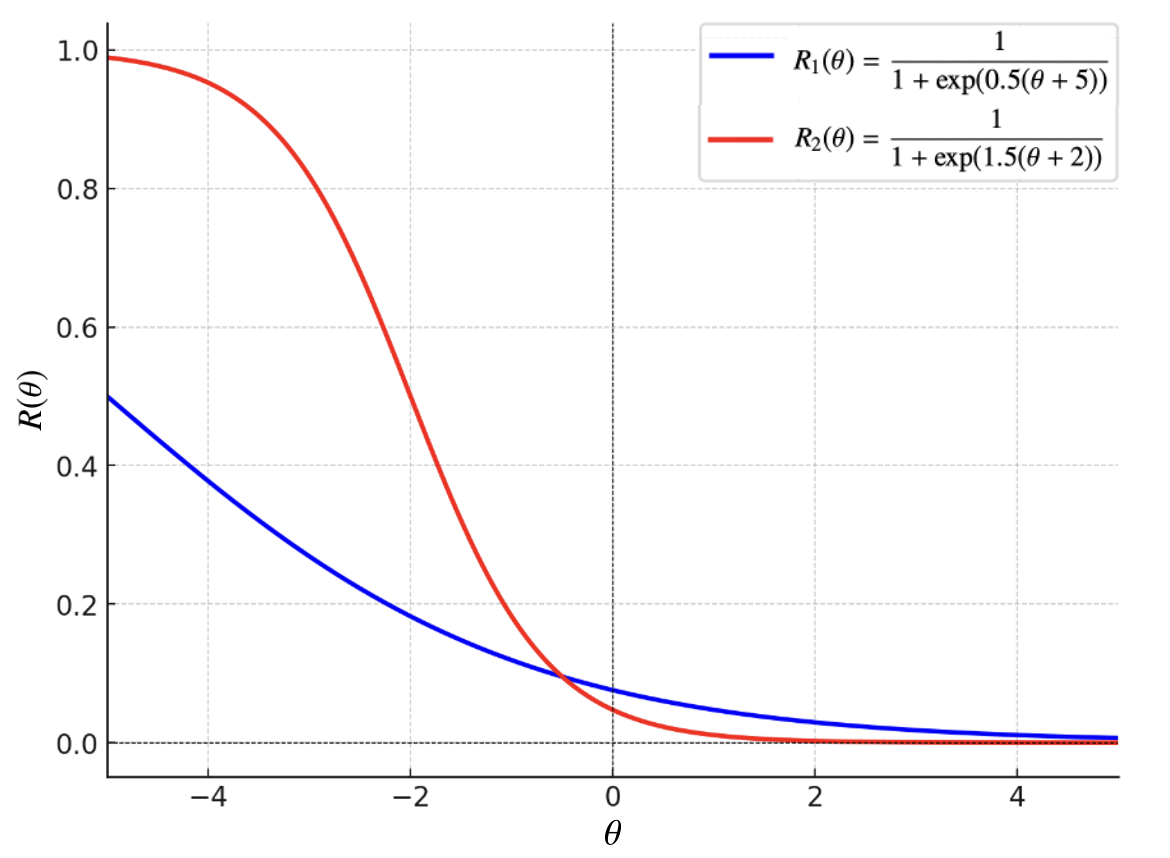} 
\caption{Ability-dependent probabilities of guessing $R(\theta)$}
\label{fig:r_function} 
\end{figure}
Table \ref{tab:careless_response} shows the number of guessers and the total guessed responses observed under each random response generation mechanism when $I=1000$ and $J=15$. For the ability-dependent mechanisms $R_1$ and $R_2$, the results are averaged over 1000 repetitions.

\begin{table}[htbp]
\centering
\setlength{\tabcolsep}{8pt}
\renewcommand{\arraystretch}{1.2}
\caption{Number of guessers and total guessed responses under each condition ($I=1000,\, J=15$)}
\vspace{2mm}
\begin{tabular}{ccc}
\toprule
Guessing Mechanism & Number of Guessers & Total Guessed Responses \\
\midrule
$(\text{Prev.}, \text{Sev.})=(0.1, 0.3)$ & 100 & 450 \\
$(\text{Prev.}, \text{Sev.})=(0.1, 0.5)$ & 100 & 750 \\
$(\text{Prev.}, \text{Sev.})=(0.3, 0.3)$ & 300 & 1350 \\
$(\text{Prev.}, \text{Sev.})=(0.3, 0.5)$ & 300 & 2250 \\
$R_1(\theta)=\frac{1}{1+\exp(0.5(\theta+5))}$ & 954 & 1273 \\
$R_2(\theta)=\frac{1}{1+\exp(1.5(\theta+2))}$ & 585 & 1457 \\
\bottomrule
\end{tabular}
\label{tab:careless_response}\\
\vspace{2mm}
The ability-dependent guessing $(R_1, R_2)$ results represent averages over 1000 repetitions.
\end{table}

Scenario 3 models an ideal condition without guessing, where all responses are consistent with the 1PLM. The estimation methods are assessed under conditions where $I=500,1000,2000,5000$ and $J=15,30,45$.

In each scenario, the item difficulty parameters are estimated using the robust DPD- and $\gamma$-divergence-based methods, standard MMLE and robust MMLE proposed by Hong and Cheng (\cite{hong2019}). Estimation is performed in the R environment (R Core Team, \cite{r2024}). To implement the proposed methods, the Gauss-Hermite quadrature was configured with 21 nodes, and the initial parameter values were set to $\bm{b}^{(0)}=\bm{0}$. The convergence criterion was defined as
\[
\max_{j} |b_j^{(t+1)}-b_j^{(t)}| < 0.0001.
\]
For the implementation of both standard MMLE and robust MMLE methods, the `mirt' package (Chalmers, \cite{chalmers2012}) was utilized, with default settings applied for all detailed configurations. 

We employ two metrics to evaluate estimation accuracy: bias and root-mean-square error (RMSE):
\[
\text{Bias} = \frac{1}{J} \sum_{j=1}^J \hat{b}_j - b_j, \quad 
\text{RMSE} = \sqrt{\frac{1}{J} \sum_{j=1}^J (\hat{b}_j - b_j)^2}.
\]
We repeated the estimation process 1,000 times for each case with each method and reported the average values of bias and RMSE. In comparing estimation methods, we emphasize RMSE as the primary metric because it provides a comprehensive measure of prediction accuracy by accounting for both the estimates' variance and bias.

\subsubsection{Results and Discussion}\label{Results and Discussion}
First, the results for Scenario 1---where guessing probability was uniform across the ability range---were displayed in Tables~\ref{random} and \ref{biased}. Table~\ref{random} summarized the Unbiased-type responses, while Table~\ref{biased} focused on the Biased-type responses. Overall, the proposed robust estimation methods effectively mitigated the adverse impact of guessing on estimation accuracy, demonstrating superior performance compared to both standard MMLE and robust MMLE proposed by Hong and Cheng (\cite{hong2019}).

\begin{table}[hbtp]
\centering
\fontsize{9pt}{10.5pt}\selectfont
\setlength{\tabcolsep}{6.1pt}
\renewcommand{\arraystretch}{1.2}
\caption{Bias and RMSE of each method in Unbiased type of Scenario 1 (the probability of guessing is uniform with respect to $\theta$): averages over 1000 iterations}
\vspace{1mm}
\begin{tabular}{rrrr|rrrrrrrrr}
\toprule
\multicolumn{4}{c}{Settings} & &  \raisebox{-1ex}{\multirow{2}{*}{MMLE}} &
\multicolumn{3}{c}{DPD} & \multicolumn{3}{c}{$\gamma$-$\text{divergence}$} & \raisebox{-1ex}{\multirow{2}{*}{\makecell{\vspace{0.2mm} Robust \\ MMLE}}} \\
\cmidrule(lr){1-4}  \cmidrule(lr){7-9} \cmidrule(lr){10-12}  
\multicolumn{1}{c}{Prev.} & \multicolumn{1}{c}{Sev.} & \multicolumn{1}{c}{$I$} & \multicolumn{1}{c}{$J$} &  &  & \multicolumn{1}{c}{0.1} & \multicolumn{1}{c}{0.3} & \multicolumn{1}{c}{0.5} &  \multicolumn{1}{c}{0.1} & \multicolumn{1}{c}{0.3} & \multicolumn{1}{c}{0.5} &   \\
\midrule
0.1 & 0.3 & 500 & 15 & Bias  & \textbf{0.004} & \textbf{0.004} & \textbf{0.004} & \textbf{0.004} & \textbf{0.004} & \textbf{0.004} & \textbf{0.004} & \textbf{0.004}   \\
\rowcolor{mygray!80} &  &  &  & RMSE & 0.125 & 0.094 & 0.097 & 0.115 & \textbf{0.093} & 0.096 & 0.119 & 0.125    \\
 &  &  & 30 & Bias  & 0.003 & \textbf{0.001} & $-0.006$ & $-0.013$ & \textbf{0.001} & $-0.006$ & $-0.014$ & $-0.002$  \\
\rowcolor{mygray!80}  & & & & RMSE & 0.129 & 0.095 & 0.129 & 0.208 & \textbf{0.094} & 0.133 & 0.215 & 0.236 \\
 &  & 1000 & 15 & Bias & $-0.003$ & $-0.004$ & $-0.003$ & $\bm{-0.001}$ & $-.0004$ & $-0.003$ & $-0.002$ & $-0.006$  \\
\rowcolor{mygray!80} &  &  &  & RMSE & 0.117 & 0.075 & 0.070 & 0.081 & 0.075 & \textbf{0.068} & 0.083 & 0.332   \\
 & &  & 30 & Bias & $-0.003$ & $-0.003$ & \textbf{0.002} & 0.006 & $-0.003$ & \textbf{0.002} & 0.007 & $-0.008$   \\
\rowcolor{mygray!80} & &  &  & RMSE & 0.120 & 0.073 & 0.089 & 0.137 & \textbf{0.072} & 0.091 & 0.141 & 0.218   \\
\hline
 & 0.5 & 500 & 15 & Bias  & 0.002 & 0.001 & \textbf{0.000} & $-0.001$ & 0.001 & \textbf{0.000} & $-0.002$ & 0.001   \\
\rowcolor{mygray!80} &  &  &  & RMSE & 0.198 & 0.113 & 0.100 & 0.117 & 0.112 & \textbf{0.098} & 0.121 & 0.308    \\
 &  &  & 30 & Bias & 0.001 & 0.002 & $-0.002$ & $-0.006$ & 0.002 & $-0.002$ & $-0.007$ & \textbf{0.000}   \\
\rowcolor{mygray!80} &  &  &  & RMSE & 0.178 & 0.099 & 0.128 & 0.204 & \textbf{0.097} & 0.131 & 0.210 & 0.219  \\ 
 &  & 1000 & 15 & Bias & $-0.005$ & $-.0005$ & $-0.003$ & $\bm{-0.001}$ & $-0.005$ & $-0.004$ & $-0.002$ & $-0.008$  \\
\rowcolor{mygray!80} &  &  &  & RMSE & 0.194 & 0.097 & 0.071 & 0.081 & 0.095 & \textbf{0.069} & 0.084 & 0.294   \\
 & &  & 30 & Bias & $-0.005$ & $-0.004$ & \textbf{0.000} & 0.003 & $-0.004$ & \textbf{0.000} & 0.004 & $-0.009$   \\
\rowcolor{mygray!80} & &  &  & RMSE & 0.175 & 0.078 & 0.089 & 0.137 & \textbf{0.076} & 0.092 & 0.141 & 0.202   \\
\hline
0.3 & 0.3 & 500 & 15 & Bias  & \textbf{0.001} & $-\bm{0.001}$ & $-0.006$ & $-0.009$ & $-0.002$ & $-0.007$ & $-0.011$ & $-0.005$   \\
\rowcolor{mygray!80} &  &  &  & RMSE & 0.275 & 0.181 & 0.132 & 0.130 & 0.178 & \textbf{0.118} & 0.128 & 0.189    \\
 &  &  & 30 & Bias & \textbf{0.000} & $-0.005$ & $-0.012$ & $-0.017$ & $-0.005$ & $-0.016$ & $-0.023$ & $-0.007$   \\
\rowcolor{mygray!80} &  &  &  & RMSE & 0.283 & 0.171 & 0.141 & 0.223 & 0.165 & 0.145 & 0.248 & \textbf{0.125}  \\ 
 &  & 1000 & 15 & Bias & $-0.007$ & $-0.009$ & $-0.009$ & $\bm{-0.006}$ & $-0.009$ & $-0.010$ & $-0.008$ & $-0.017$  \\
\rowcolor{mygray!80} &  &  &  & RMSE & 0.274 & 0.172 & 0.110 & 0.096 & 0.169 & 0.092 & \textbf{0.091} & 0.173   \\
 & &  & 30 & Bias & $-0.008$ & $-0.010$ & \textbf{0.002} & 0.016 & $-0.010$ & 0.003 & 0.024 & $-0.023$   \\
\rowcolor{mygray!80} & &  &  & RMSE & 0.282 & 0.159 & \textbf{0.099} & 0.155  & 0.153 & 0.100 & 0.174 & 0.102   \\
\hline
 & 0.5 & 500 & 15 & Bias & $\bm{-0.002}$ & $-0.005$ & $-0.009$ & $-0.011$ & $-0.005$ & $-0.011$ & $-0.015$ & $-0.008$  \\
\rowcolor{mygray!80} &  &  &  & RMSE & 0.442 & 0.277 & 0.154 & 0.140 & 0.272 & 0.130 & 0.136 & \textbf{0.119}    \\
 &  &  & 30 & Bias & $\bm{-0.002}$ & $-0.006$ & $-0.013$ & $-0.017$ & $-0.007$ & $-0.018$ & $-0.023$ & $-0.008$   \\
\rowcolor{mygray!80} &  &  &  & RMSE & 0.400 & 0.201 & 0.142 & 0.224 & 0.193 & 0.149 & 0.249 & \textbf{0.107}  \\ 
 &  & 1000 & 15 & Bias & $-0.011$ & $-0.014$ & $-0.013$ & $\bm{-0.007}$ & $-0.015$ & $-0.016$ & $-0.011$ & $-0.025$  \\
\rowcolor{mygray!80} &  &  &  & RMSE & 0.441 & 0.269 & 0.133 & 0.106 & 0.264 & 0.102 & 0.098 & \textbf{0.096}   \\
 & &  & 30 & Bias & $-0.011$ & $-0.013$ & \textbf{0.001} & 0.016 & $-0.014$ & 0.002 & 0.025 & $-0.029$   \\
\rowcolor{mygray!80} & &  &  & RMSE & 0.401 & 0.191 & 0.101 & 0.159 & 0.181 & 0.106 & 0.181 & \textbf{0.085}   \\
\bottomrule
\end{tabular}
\begin{flushleft}
\begingroup
\fontsize{9pt}{10.5pt}\selectfont
$\ \bullet \ \ \ $Prev.: Prevalence,\,\, Sev.: Severity,\,\, $I$: Sample size,\,\, $J$: Test length\\
$\ \bullet \ \ \ $MMLE: Marginal maximum likelihood estimation.\\
$\ \bullet \ \ \ $DPD and $\gamma$-divergence: The proposed methods (the second row of the table provides the hyperparameter values).\\
$\ \bullet \ \ \ $Robust MMLE: Robust marginal maximum likelihood estimation by Hong and Cheng (\cite{hong2019}).\\[2pt]
The bias and RMSE values closest to 0 are highlighted in bold for each scenario.
\endgroup
\end{flushleft}
\label{random}
\end{table}

\begin{table}[hbtp]
\fontsize{9pt}{10.5pt}\selectfont
\centering
\setlength{\tabcolsep}{6.3pt}
\renewcommand{\arraystretch}{1.2}
\caption{Bias and RMSE of each method in Biased type of Scenario 1 (the probability of guessing is uniform with respect to $\theta$): averages over 1000 iterations}
\vspace{1mm}
\begin{tabular}{cccc|ccccccccc}
\toprule
\multicolumn{4}{c}{Settings} & &  \raisebox{-1ex}{\multirow{2}{*}{MMLE}} &
\multicolumn{3}{c}{ DPD} & \multicolumn{3}{c}{$\gamma$-$\text{divergence}$} & \raisebox{-1ex}{\multirow{2}{*}{\makecell{\vspace{0.2mm} Robust \\ MMLE}}} \\
\cmidrule(lr){1-4}  \cmidrule(lr){7-9} \cmidrule(lr){10-12}  
\multicolumn{1}{c}{Prev.} & \multicolumn{1}{c}{Sev.} & \multicolumn{1}{c}{$I$} & \multicolumn{1}{c}{$J$} &  &  & \multicolumn{1}{c}{\hspace{0.5mm} 0.1} & \multicolumn{1}{c}{\hspace{0.5mm} 0.3} & \multicolumn{1}{c}{0.5} &  \multicolumn{1}{c}{\hspace{1mm} 0.1} & \multicolumn{1}{c}{\hspace{0.5mm} 0.3} & \multicolumn{1}{c}{0.5} &   \\
\midrule
0.1 & 0.3 & 500 & 15 & Bias & 0.045 & \hspace{0.5mm} 0.043 & \hspace{0.5mm} 0.036 & \textbf{0.030} & \hspace{1mm}  0.043 & \hspace{0.5mm} 0.037 & 0.031 & 0.055   \\
\rowcolor{mygray!80} &  &  &  & RMSE & 0.129 & \textbf{0.103} & \textbf{0.103} & 0.119 & \textbf{0.103} & \textbf{0.103} & 0.122 & 0.354    \\
 &  &  & 30 & Bias  & 0.047 & 0.035 & 0.011 & $\bm{-0.002}$ & 0.035 & 0.012 & $-0.003$ & 0.040   \\
\rowcolor{mygray!80} &  &  &  & RMSE & 0.130 & 0.101 & 0.129 & 0.206 & \textbf{0.100} & 0.132 & 0.212 & 0.244    \\
 &  & 1000 & 15 & Bias  & 0.038 & 0.037 & 0.033 & \textbf{0.029} & 0.037 & 0.034 & 0.031 & 0.047   \\
\rowcolor{mygray!80} &  &  &  & RMSE & 0.117 & 0.083 & 0.078 & 0.086 & 0.083 & \textbf{0.077} & 0.088 & 0.340    \\
 &  &  & 30 & Bias  & 0.041 & 0.032 & 0.020 & \textbf{0.018} & 0.033 & 0.022 & 0.019 & 0.035   \\
\rowcolor{mygray!80} &  &  &  & RMSE & 0.119 & 0.080 & 0.091 & 0.137 & \textbf{0.079} & 0.092 & 0.141 & 0.225    \\
\hline
 & 0.5 & 500 & 15 & Bias  & 0.079 & 0.062 & 0.039 & \textbf{0.026} & 0.062 & 0.040 & 0.028 & 0.059   \\
\rowcolor{mygray!80} &  &  &  & RMSE & 0.191 & 0.129 & 0.109 & 0.121 & 0.128 & \textbf{0.108} & 0.123 & 0.325    \\
 &  &  & 30 & Bias  & 0.072 & 0.040 & 0.010 & \textbf{0.000} & 0.041 & 0.011 & \textbf{0.000} & 0.040   \\
\rowcolor{mygray!80} &  &  &  & RMSE & 0.174 & 0.108 & 0.128 & 0.203 & \textbf{0.107} & 0.130 & 0.209 & 0.230    \\
 &  & 1000 & 15 & Bias  & 0.073 & 0.058 & 0.038 & \textbf{0.030} & 0.058 & 0.040 & 0.032 & 0.052   \\
\rowcolor{mygray!80} &  &  &  & RMSE & 0.182 & 0.113 & 0.085 & 0.088 & 0.112 & \textbf{0.083} & 0.089 & 0.313    \\
 &  &  & 30 & Bias  & 0.066 & 0.036 & 0.014 & \textbf{0.010} & 0.037 & 0.015 & 0.011 & 0.031   \\
\rowcolor{mygray!80} &  &  &  & RMSE & 0.166 & 0.088 & 0.091 & 0.137 & \textbf{0.087} & 0.092 & 0.141 & 0.212    \\
\hline
0.3 & 0.3 & 500 & 15 & Bias  & 0.109 & 0.108 & 0.087 & \textbf{0.068} & 0.109 & 0.094 & 0.079 & 0.141   \\
\rowcolor{mygray!80} &  &  &  & RMSE & 0.282 & 0.205 & 0.159 & \textbf{0.149} & 0.204 & 0.154 & 0.150 & 0.251    \\
 &  &  & 30 & Bias  & 0.117 & 0.091 & 0.039 & \textbf{0.019} & 0.094 & 0.045 & 0.021 & 0.120   \\
\rowcolor{mygray!80} &  &  &  & RMSE & 0.286 & 0.187 & \textbf{0.147} & 0.220 & 0.185 & 0.150 & 0.239 & 0.183    \\
 &  & 1000 & 15 & Bias  & 0.102 & 0.104 & 0.092 & \textbf{0.080} & 0.105 & 0.099 & 0.094 & 0.133   \\
\rowcolor{mygray!80} &  &  &  & RMSE & 0.273 & 0.193 & 0.144 & \textbf{0.128} & 0.191 & 0.139 & 0.131 & 0.238    \\
 &  &  & 30 & Bias  & 0.109 & 0.091 & 0.058 & \textbf{0.054} & 0.094 & 0.071 & 0.071 & 0.107   \\
\rowcolor{mygray!80} &  &  &  & RMSE & 0.279 & 0.176 & \textbf{0.118} & 0.160 & 0.173 & 0.122 & 0.179 & 0.159    \\
\hline
 & 0.5 & 500 & 15 & Bias  & 0.191 & 0.162 & 0.106 & \textbf{0.073} & 0.164 & 0.117 & 0.087 & 0.169   \\
\rowcolor{mygray!80} &  &  &  & RMSE & 0.434 & 0.306 & 0.201 & 0.168 & 0.305 & 0.192 & \textbf{0.165} & 0.228   \\
 &  &  & 30 & Bias  & 0.177 & 0.106 & 0.026 & \textbf{0.003} & 0.110 & 0.030 & \textbf{0.003} & 0.115   \\
\rowcolor{mygray!80} &  &  &  & RMSE & 0.396 & 0.228 & \textbf{0.148} & 0.221 & 0.224 & 0.149 & 0.242 & 0.167    \\
 &  & 1000 & 15 & Bias  & 0.182 & 0.156 & 0.108 & \textbf{0.081} & 0.158 & 0.118 & 0.098 & 0.155   \\
\rowcolor{mygray!80} &  &  &  & RMSE & 0.426 & 0.296 & 0.188 & 0.148 & 0.294 & 0.178 & \textbf{0.145} & 0.208    \\
 &  &  & 30 & Bias  & 0.167 & 0.103 & 0.043 & \textbf{0.037} & 0.107 & 0.054 & 0.052 & 0.096   \\
\rowcolor{mygray!80} &  &  &  & RMSE & 0.389 & 0.218 & \textbf{0.117} & 0.161 & 0.214 & 0.118 & 0.181 & 0.137    \\
\bottomrule
\end{tabular}
\begin{flushleft}
\begingroup
\fontsize{9pt}{10.5pt}\selectfont
$\ \bullet \ \ \ $Prev.: Prevalence,\,\, Sev.: Severity,\,\, $I$: Sample size,\,\, $J$: Test length\\
$\ \bullet \ \ \ $MMLE: Marginal maximum likelihood estimation.\\
$\ \bullet \ \ \ $DPD and $\gamma$-divergence: The proposed methods (the second row of the table provides the hyperparameter values).\\
$\ \bullet \ \ \ $Robust MMLE: Robust marginal maximum likelihood estimation by Hong and Cheng (\cite{hong2019}).\\[2pt]
The bias and RMSE values closest to 0 are highlighted in bold for each scenario.
\endgroup
\end{flushleft}
\label{biased}
\end{table}

As shown in Table \ref{random}, in scenarios where Unbiased-type guessing responses occurred uniformly with respect to $\theta$, all methods exhibited consistently low bias. Under these conditions, the proposed DPD and $\gamma$-divergence estimators achieved the lowest RMSE in almost all cases when the aberrant response rate was relatively low. Even with a higher proportion of aberrant responses, adjusting $\beta$ or $\gamma$ to larger values prevented substantial RMSE deterioration for these methods. In contrast, for $(\text{Prev.}, \text{Sev.})=(0.3, 0.5)$, robust MMLE attained the most favorable outcome. In addition, standard MMLE suffered a noticeable RMSE increase as the proportion of aberrant responses rose, resulting in a pronounced decline in estimation accuracy.

In scenarios characterized by a uniform occurrence of Biased-type guessing responses with respect to $\theta$ (Table~\ref{biased}), the benefits of the proposed methods stood out even more distinctly than in the Unbiased-type setting. Specifically, both DPD and $\gamma$-divergence estimators consistently surpassed existing approaches in terms of bias and RMSE across the range of examined conditions. It was also observed that, compared to Unbiased-type responses, selecting slightly larger values for 
$\beta$ or $\gamma$ enhanced the estimation accuracy for Biased-type guessing responses. On the other hand, MMLE experienced pronounced increases in bias and RMSE as the proportion of aberrant responses rose, causing a notable drop in estimation accuracy. Although robust MMLE mitigated these aberrant effects to some degree relative to MMLE, its effectiveness was limited compared with the proposed methods.

Next, the results of Scenario 2, where the probability of guessing responses depended on $\theta$, were presented in Table~\ref{depend}. While it was evident that the proposed methods and robust MMLE outperformed standard MMLE in terms of RMSE in this scenario, identifying a superior approach remained difficult. Under mechanism $R_1$, robust MMLE performed particularly well, especially for the Unbiased-type responses. Conversely, under mechanism $R_2$, the proposed methods with smaller hyperparameters $\beta$ or $\gamma$ often achieved the highest estimation accuracy. Integrating these results with those from Scenario 1 suggests that when the proportion of respondents exhibiting aberrant responses is relatively small, the proposed methods are more likely to achieve improved estimation accuracy than robust MMLE.

\begin{table}[hbtp]
\centering
\fontsize{9pt}{10.5pt}\selectfont
\setlength{\tabcolsep}{5.5pt}
\renewcommand{\arraystretch}{1.2}
\caption{Bias and RMSE of each method in Scenario 2 (the probability of guessing depends on $\theta$): averages over 1000 iterations}
\vspace{1mm}
\begin{tabular}{c@{\hskip 3pt}r@{\hskip 3pt}rr|rrrrrrrrr}
\toprule
\multicolumn{4}{c}{Settings} & &  \raisebox{-1ex}{\multirow{2}{*}{MMLE}} &
\multicolumn{3}{c}{DPD} & \multicolumn{3}{c}{$\gamma$-$\text{divergence}$} & \raisebox{-1ex}{\multirow{2}{*}{\makecell{\vspace{0.2mm} Robust \\ MMLE}}} \\
\cmidrule(lr){1-4}  \cmidrule(lr){7-9} \cmidrule(lr){10-12}  
\multicolumn{1}{c}{Type} & \multicolumn{1}{c}{Mech.} & \multicolumn{1}{c}{$I$} & \multicolumn{1}{c}{$J$} &  &  & \multicolumn{1}{c}{0.1} & \multicolumn{1}{c}{0.3} & \multicolumn{1}{c}{0.5} &  \multicolumn{1}{c}{0.1} & \multicolumn{1}{c}{0.3} & \multicolumn{1}{c}{0.5} &   \\
\midrule
Unbiased & \multicolumn{1}{c}{$R_1$} & 500 & 15 & Bias  & $\bm{-0.026}$ & $-0.057$ & $-.0097$ & $-0.121$ & $-0.058$ & $-0.111$ & $-0.165$ & $-0.104$   \\
\rowcolor{mygray!80} & &  &  & RMSE & 0.297 & 0.222 & 0.196 & 0.194 & 0.219 & 0.184 & 0.214 & \textbf{0.160}    \\
 & & & 30 & Bias  & $\bm{-0.031}$ & $-0.083$ & $-0.119$ & $-0.146$ & $-0.087$ & $-0.194$ & $-0.345$ & $-0.129$   \\
\rowcolor{mygray!80} & &  &  & RMSE & 0.275 & 0.229 & 0.198 & 0.285 & 0.225 & 0.242 & 0.523 & \textbf{0.166}   \\
 & & 1000 & 15 & Bias  & $\bm{-0.031}$ & $-0.066$ & $-0.110$ & $-0.137$ & $-0.067$ & $-0.126$ & $-0.184$ & $-0.113$   \\
\rowcolor{mygray!80} & &  &  & RMSE & 0.293 & 0.219 & 0.195 & 0.190 & 0.216 & 0.184 & 0.210 & \textbf{0.148}    \\
 & &  & 30 & Bias  & $\bm{-0.034}$ & $-0.090$ & $-0.131$ & $-0.161$ & $-0.095$ & $-0.210$ & $-0.355$ & $-0.131$   \\
\rowcolor{mygray!80} & &  &  & RMSE & 0.273 & 0.228 & 0.191 & 0.230 & 0.224 & 0.238 & 0.419 & \textbf{0.153}   \\
\hline
& \multicolumn{1}{c}{$R_2$} & 500 & 15 & Bias  & $\bm{-0.085}$ & $-0.148$ & -0.207 & $-0.226$ & $-0.151$ & -0.244 & $-0.327$ & $-0.244$   \\
\rowcolor{mygray!80} & &  &  & RMSE & 0.322 & 0.238 & 0.246 & 0.263 & \textbf{0.237} & 0.269 & 0.363 & 0.288    \\
 & &  & 30 & Bias  & $\bm{-0.086}$ & $-0.185$ & $-0.236$ & $-0.263$ & $-0.197$ & $-0.388$ & $-0.616$ & $-0.265$   \\
\rowcolor{mygray!80} & &  &  & RMSE & 0.316 & \textbf{0.263} & 0.276 & 0.338 & 0.266 & 0.413 & 0.719 & 0.285   \\
 & & 1000 & 15 & Bias  & $\bm{-0.088}$ & $-0.153$ & $-0.216$ & $-0.238$ & $-0.156$ & $-0.253$ & $-0.339$ & $-0.247$   \\
\rowcolor{mygray!80} & &  &  & RMSE & 0.319 & 0.237 & 0.245 & 0.261 & \textbf{0.236} & 0.269 & 0.358 & 0.279    \\
 & &  & 30 & Bias  & $\bm{-0.090}$ & $-0.193$  & $-0.250$ & $-0.277$ & $-0.205$ & $-0.402$ & $-0.611$ & $-0.267$   \\
\rowcolor{mygray!80} & &  &  & RMSE & 0.312 & \textbf{0.262} & 0.276 & 0.314 & 0.266 & 0.402 & 0.652 & 0.277   \\
\hline
Biased & \multicolumn{1}{c}{$R_1$} & 500 & 15 & Bias  & 0.091 & 0.089 & 0.068 & \textbf{0.042} & 0.090 & 0.074 & 0.052 & 0.130   \\
\rowcolor{mygray!80} & &  &  & RMSE & 0.279 & 0.224 & 0.190 & 0.174 & 0.222 & 0.178 & \textbf{0.161} & 0.196    \\
 & &  & 30 & Bias  & 0.082 & 0.068 & 0.043 & \textbf{0.038} & 0.070 & 0.058 & 0.065 & 0.109   \\
\rowcolor{mygray!80} & &  &  & RMSE & 0.253 & 0.213 & 0.170 & 0.245 & 0.209 & 0.167 & 0.306 & \textbf{0.156}   \\
 & & 1000 & 15 & Bias  & 0.084 & 0.077 & 0.048 & \textbf{0.019} & 0.078 & 0.053 & 0.023 & 0.114   \\
\rowcolor{mygray!80} & &  &  & RMSE & 0.270 & 0.213 & 0.176 & 0.155 & 0.211 & 0.161 & \textbf{0.133} & 0.170    \\
 & &  & 30 & Bias  & 0.078 & 0.060 & 0.028 & \textbf{0.016} & 0.062 & 0.037 & 0.022 & 0.105   \\
\rowcolor{mygray!80} & &  &  & RMSE & 0.248 & 0.205 & 0.148 & 0.163 & 0.201 & \textbf{0.135} & 0.187 & \textbf{0.135}   \\
\hline
 & \multicolumn{1}{c}{$R_2$} & 500 & 15 & Bias & 0.037 & $\bm{-0.001}$ & $-0.048$ & $-0.076$ & $\bm{-0.001}$ & $-0.054$ & $-0.092$ & $-0.055$  \\
\rowcolor{mygray!80} & &  &  & RMSE & 0.227 & 0.174 & 0.157 & 0.165 & 0.172 & \textbf{0.151} & 0.169 & 0.230   \\
 &  &  & 30 & Bias  & \textbf{0.039} & $-0.045$ & $-0.139$ & $-0.196$ & $-0.047$ & $-0.187$ & $-0.315$ & $-0.085$   \\
\rowcolor{mygray!80} & &  &  & RMSE & 0.228 & 0.186 & 0.206 & 0.307 & 0.183 & 0.237 & 0.431 & \textbf{0.157}   \\
 & & 1000 & 15 & Bias  & 0.032 & $\bm{-0.008}$ & $-0.060$ & $-0.092$ & $\bm{-0.008}$ & $-0.066$ & $-0.109$ & $-0.062$   \\
\rowcolor{mygray!80} & &  &  & RMSE & 0.222 & 0.167 & 0.150 & 0.153 & 0.165 & \textbf{0.142} & 0.155 & 0.215    \\
 & &  & 30 & Bias  & \textbf{0.033} & $-0.056$ & $-0.155$ & $-0.212$ & $-0.058$ & $-0.206$ & $-0.334$ & $-0.091$   \\
\rowcolor{mygray!80} & &  &  & RMSE & 0.222 & 0.180 & 0.201 & 0.263 & 0.178 & 0.236 & 0.379 & \textbf{0.142}   \\
\bottomrule
\end{tabular}
\begin{flushleft}
\begingroup
\fontsize{9pt}{10.5pt}\selectfont
$\ \bullet \ \ \ $Type: Unbiased: $P(u=1)=P(u=0)=0.5$,\,\, Biased: $P(u=1)=0.8, \ P(u=0)=0.2$. \\
$\ \bullet \ \ \ $Mech.: The probability of generating aberrant responses for each responder; \\
\centering
\vspace{1mm}
$R_1(\theta)=1/(1+\exp(0.5(\theta+5)))$, \quad $R_2(\theta)=1/(1+\exp(1.5(\theta+2)))$.\\
\vspace{1mm}
\raggedright
$\ \bullet \ \ \ $$I$: Sample size,\,\, $J$: Test length\\
$\ \bullet \ \ \ $MMLE: Marginal maximum likelihood estimation.\\
$\ \bullet \ \ \ $DPD, $\gamma$-divergence: The proposed methods (the second row of the table provides the hyperparameter values). \\
$\ \bullet \ \ \ $Robust MMLE: Robust marginal maximum likelihood estimation by Hong and Cheng (\cite{hong2019}).\\[2pt]
The bias and RMSE values closest to 0 are highlighted in bold for each scenario.
\endgroup
\end{flushleft}
\label{depend}
\end{table}

Table \ref{no} highlighted that standard MMLE achieved high efficiency without aberrant responses. Similarly, the proposed robust estimation methods performed comparably to MMLE when small hyperparameter values were selected, whereas higher hyperparameter settings reduced the proposed methods' accuracy. This aligns with the theoretical consistency that DPD and $\gamma$-divergence converge to KL divergence as the hyperparameters approach zero. The results also suggested that the estimation accuracy of robust MMLE, designed to accommodate aberrant responses, is more likely to deteriorate in scenarios with few or no aberrant responses than the other estimation methods.

\begin{table}[hbtp]
\fontsize{9pt}{10.5pt}\selectfont
\setlength{\tabcolsep}{7.8pt}
\centering
\renewcommand{\arraystretch}{1.2}
\caption{Bias and RMSE of each method in Scenario 3 (no guessing): averages over 1000 iterations}
\vspace{1mm}
\begin{tabular}{rr|rrrrrrrrr}
\toprule
\multicolumn{2}{c}{Settings} & &  \raisebox{-1ex}{\multirow{2}{*}{MMLE}} &
\multicolumn{3}{c}{DPD} & \multicolumn{3}{c}{$\gamma$-$\text{divergence}$} & \raisebox{-1ex}{\multirow{2}{*}{ \makecell{\vspace{0.2mm} Robust \\ MMLE}}} \\
\cmidrule(lr){1-2}  \cmidrule(lr){5-7} \cmidrule(lr){8-10}  
 \multicolumn{1}{c}{$I$} & \multicolumn{1}{c}{$J$} &  &  & \multicolumn{1}{c}{0.1} & \multicolumn{1}{c}{0.3} & \multicolumn{1}{c}{0.5} &  \multicolumn{1}{c}{0.1} & \multicolumn{1}{c}{0.3} & \multicolumn{1}{c}{0.5} &   \\
\midrule
 500 & 15 & Bias  & \textbf{0.004} & 0.005 & 0.008 & 0.010 & 0.005 & 0.008 & 0.011 & 0.008   \\
\rowcolor{mygray!80}  &  & RMSE & \textbf{0.082} & 0.083 & 0.095 & 0.114 & 0.083  & 0.096 & 0.115 & 0.424   \\
  & 30 & Bias  & 0.004 & \textbf{0.001} & $-0.006$ & $-0.013$ & \textbf{0.001} & $-0.006$ & $-0.014$ & 0.006   \\
\rowcolor{mygray!80}  &  & RMSE & \textbf{0.080} & 0.095 & 0.129 & 0.208 & 0.094 & 0.133 & 0.215 & 0.294 \\
  & 45 & Bias  & \textbf{0.004} & 0.011 & 0.012 & 0.009 & 0.011 & 0.013 & 0.010 & \textbf{0.004}  \\
\rowcolor{mygray!80}  &  & RMSE & \textbf{0.080} & 0.094 & 0.185 & 0.412 & 0.094 & 0.186 & 0.444 & 0.244   \\
 1000 & 15 & Bias  & $\bm{-0.002}$ & $-0.003$ & $-0.005$ & $-0.007$ & $-0.003$ & $-0.005$ & $-0.007$ & $-0.005$   \\
\rowcolor{mygray!80}  &  & RMSE & \textbf{0.058} & \textbf{0.058} & 0.067 & 0.080 & \textbf{0.058} & 0.067 & 0.081 & 0.411   \\
  & 30 & Bias  & $\bm{-0.001}$ & $-0.002$ & $-0.006$ & $-0.007$ & $-0.002$ & $-0.006$ & $-0.007$ & $-0.002$   \\
\rowcolor{mygray!80}  &  & RMSE & \textbf{0.057} & 0.060 & 0.086 & 0.133 & 0.060 & 0.086 & 0.133 & 0.281  \\
  & 45 & Bias  & $-0.002$ & \textbf{0.000} & $-0.004$ & $-0.005$ & \textbf{0.000} & $-0.004$ & $-0.006$ & $-0.002$   \\
\rowcolor{mygray!80}  &  & RMSE & \textbf{0.057} & 0.066 & 0.126 & 0.280 & 0.066 & 0.126 & 0.278 & 0.229   \\
 2000 & 15 & Bias  & 0.021 & 0.020 & 0.016 & \textbf{0.013} & 0.020 & 0.017 & \textbf{0.013} & 0.024   \\
\rowcolor{mygray!80}  &  & RMSE & \textbf{0.046} & \textbf{0.046} & 0.050 & 0.059 & \textbf{0.046} & 0.051 & 0.060 & 0.404   \\
  & 30 & Bias  & 0.021 & 0.019 & 0.009 & \textbf{0.001} & 0.019 & 0.009 & $\bm{-0.001}$ & 0.022   \\
\rowcolor{mygray!80}  &  & RMSE & \textbf{0.045} & 0.046 & 0.062 & 0.094 & 0.046 & 0.062 & 0.094 & 0.276   \\
  & 45 & Bias  & 0.021 & 0.018 & $\bm{-0.001}$ & $-0.011$ & 0.018 & $-0.002$ & $-0.015$ & 0.021   \\
\rowcolor{mygray!80}  &  & RMSE & \textbf{0.045} & 0.049 & 0.089 & 0.194 & 0.049 & 0.089 & 0.187 & 0.221   \\
 5000 & 15 & Bias  & 0.020 & 0.019 & 0.014 & \textbf{0.008} & 0.019 & 0.014 & \textbf{0.008} & 0.021   \\
\rowcolor{mygray!80}  &  & RMSE & \textbf{0.032} & \textbf{0.032} & 0.033 & 0.038 & \textbf{0.032} & 0.033 & 0.038 & 0.400   \\
  & 30 & Bias  & 0.019 & 0.017 & \textbf{0.001} & $-0.010$ & 0.017 & \textbf{0.001} & -0.013 & 0.019   \\
\rowcolor{mygray!80}  &  & RMSE & \textbf{0.032} & \textbf{0.032} & 0.040 & 0.061 & \textbf{0.032} & 0.040 & 0.062 & 0.271   \\
  & 45 & Bias  & 0.019 & \textbf{0.014} & $\bm{-0.014}$ & $-0.028$ & \textbf{0.014} & $-0.017$ & $-0.035$ & 0.018   \\
\rowcolor{mygray!80}  &  & RMSE & \textbf{0.031} & 0.032 & 0.058 & 0.122 & 0.032 & 0.059 & 0.120 & 0.217   \\
\bottomrule
\end{tabular}
\begin{flushleft}
\begingroup
\fontsize{9pt}{10.5pt}\selectfont
$\bullet \ \ \ $$I$: Sample size,\,\, $J$: Test length.\\
$\bullet \ \ \ $MMLE: Marginal maximum likelihood estimation.\\
$\bullet \ \ \ $DPD and $\gamma$-divergence: The proposed methods (the second row of the table provides the hyperparameter values).\\
$\bullet \ \ \ $Robust MMLE: Robust marginal maximum likelihood estimation by Hong and Cheng (\cite{hong2019}).\\[2pt]
The bias and RMSE values closest to 0 are highlighted in bold for each scenario.
\endgroup
\end{flushleft}
\label{no}
\end{table}

In summary, under conditions involving aberrant responses, the proposed methods based on DPD and $\gamma$-divergence demonstrated more stable estimation accuracy compared to standard MMLE and robust MMLE. In particular, for Biased-type guessing responses or scenarios with high proportions of aberrant responses, the robustness of the proposed methods was particularly pronounced, often outperforming existing methods in bias and RMSE. On the other hand, depending on the mechanism and proportion of aberrant responses, robust MMLE occasionally achieved the best accuracy. A significant limitation of robust MMLE, however, lies in its reliance on ability estimates to compute weights, which makes it more susceptible to instability in ability estimation, particularly when the test length is small. In contrast, the proposed methods exhibited relatively stable estimation accuracy even with fewer items. However, their performance tended to deteriorate when the test length was large relative to the number of responders, which may be attributed to the simultaneous optimization framework employed by the proposed methods. Furthermore, even in situations with minimal aberrant responses, the proposed methods achieved accuracy comparable to MMLE by setting the hyperparameters to small values. In contrast, robust MMLE suffered from a significant reduction in estimation accuracy in such cases. Overall, the proposed methods demonstrated high robustness and flexibility, maintaining excellent estimation accuracy across various scenarios involving aberrant responses.

\subsection{Analysis Based on Influence Functions}\label{Analysis Based on Influence Functions}

This subsection investigates the robustness of the estimation procedure by examining the influence function, as defined in Eq.~\eqref{influence_function}. Specifically, we consider a scenario with $J = 5$ items whose true difficulty parameters are set to $(b_1, b_2, b_3, b_4, b_5) = (-2, -1, 0, 1, 2)$. We generate $I = 2000$ response patterns under the 1PLM, independently drawing the ability parameters from a standard normal distribution. Using the resulting sample responses $\bm{u}_1, \dots, \bm{u}_I$, we estimate the difficulty parameters via both MMLE and the proposed robust estimation methods. For each response vector $\bm{u} \in \{0, 1\}^5$, the empirical approximation of the influence function with respect to the population distribution $p$ and the estimated parameters $\bm{\hat{b}}_{\alpha}$, is given by
\[
\widehat{\text{IF}}(\bm{u}; \bm{\hat{b}}_{\alpha}, p) = -\hat{V}_{\alpha}(\bm{\hat{b}}_{\alpha})^{-1} \psi_{\alpha}(\bm{\hat{b}}_{\alpha}, \bm{u}).
\]

Table~\ref{IF} presents the average $L^2$-norms of the influence functions for each response pattern under each estimation procedure, with their occurrence probabilities ordered in descending magnitude. The last row of the table summarizes the gross-error sensitivity for each method. Several notable findings emerge from these results.

First, under MMLE, the norms of the influence functions increase considerably for response patterns that occur infrequently. In contrast, the proposed robust estimation methods effectively reduce the influence function norms for these low-probability patterns, with the effect becoming more pronounced as the hyperparameters increase. Notably, the $\gamma$-divergence approach exhibits a slightly more pronounced reduction than DPD, suggesting enhanced robustness against outliers or rare response patterns, such as incorrectly answering an easy item while correctly answering a difficult one.

An intriguing exception arises for the two extreme response patterns $(1,1,1,1,1)$ and $(0,0,0,0,0)$, wherein the influence function norms become larger under both robust methods compared to MMLE. Although these patterns could potentially be considered outliers---since they represent uniformly correct or uniformly incorrect responses---they do not conform to the conventional notion of outliers in the robust estimation framework. Instead, their influence is amplified under DPD and $\gamma$-divergence.

One possible explanation lies in how each method integrates response information into their estimators. In the MMLE framework, the difficulty of a specific item is estimated solely based on the responses associated with that item, whereby an unexpected correct response to a high-difficulty item directly yields significant influence. On the other hand, robust methods collectively incorporate information from all items, thereby capturing the overall coherence of the response pattern. When all items are answered identically---entirely correct or entirely incorrect---the responses to other items serve as a kind of contextual evidence, enhancing the plausibility of the response to a specific item. Thus, although these patterns appear extreme, their internal coherence is more strongly weighted in robust estimation, leading to an amplified influence compared to MMLE. Nevertheless, such uniform response patterns are rare, and the responses to the most difficult or manageable items typically have the most significant potential to influence the estimates. Their scarcity and heightened internal coherence result in larger influence function values for robust methods than for MMLE.

Finally, a gross-error sensitivity comparison reveals that DPD-based and $\gamma$-divergence-based estimators yield lower values than MMLE. This indicates that the proposed methods effectively suppress sensitivity to outliers (responses with low occurrence probabilities) and achieve high robustness. Excluding the uniform response uncovers a consistent drop in gross-error sensitivity as the hyperparameters rise. This demonstrates that the proposed methods allow for flexible robustness control through hyperparameter tuning. However, it is important to note that uniformly correct or uniformly incorrect responses are not automatically regarded as outliers in the robust estimation methods; instead, their high internal coherence amplifies their influence compared to MMLE.

\begin{table}[hbtp]
\fontsize{9pt}{10.5pt}\selectfont
\centering
\setlength{\tabcolsep}{8pt}
\renewcommand{\arraystretch}{1.2}
\caption{$L^2$-norms of the estimated influence function $\|\mathrm{\widehat{IF}}(\bm{u}; \bm{\hat{b}}_{\alpha}, p)\|$ for each response pattern $(I=2000)$: averages over 1000 iterations}
\vspace{1mm}
\begin{tabular}{c|c|rrrrrrr}
\toprule
\multicolumn{1}{c}{\raisebox{-0.6ex}{\multirow{2}{*}{Response}}} &\multicolumn{1}{c}{ \raisebox{-0.6ex}{\multirow{2}{*}{\hspace{-1.5mm} Prob.~(\%)}}\hspace{-1.5mm} } & \raisebox{-0.6ex}{\multirow{2}{*}{MMLE}} & \multicolumn{3}{c}{$\text{DPD}$} & \multicolumn{3}{c}{$\gamma$-$\text{divergence}$} \\
 \cmidrule(lr){4-6} \cmidrule(lr){7-9}
\multicolumn{1}{c}{} & \multicolumn{1}{c}{} & & \multicolumn{1}{c}{0.1} & \multicolumn{1}{c}{0.3} & \multicolumn{1}{c}{0.5} & \multicolumn{1}{c}{0.1} & \multicolumn{1}{c}{0.3} & \multicolumn{1}{c}{0.5} \\
\midrule
$(1,\, 1,\, 0,\, 0,\, 0)$ & 23.637 & 2.552 & 2.679 & 2.883 & 3.049 & 2.680 & 2.887 & 3.053 \\ 
$(1,\, 1,\, 1,\, 0,\, 0)$ & 23.637 & 2.582 & 2.711 & 2.921 & 3.095 & 2.712 & 2.925 & 3.098 \\ 
\rowcolor{mygray!80} $(1,\, 0,\, 0,\, 0,\, 0)$ & 13.059 & 4.155 & 4.331 & 4.697 & 5.099 & 4.329 & 4.814 & 5.073 \\ 
\rowcolor{mygray!80} $(1,\, 1,\, 1,\, 1,\, 0)$ & 13.059 & 4.281 & 4.418 & 4.855 & 5.280 & 4.466 & 4.846 & 5.256 \\ 
$(1,\, 0,\, 1,\, 0,\, 0)$ & 4.309 & 4.242 & 4.027 & 3.510 & 2.968 & 4.025 & 3.487 & 2.893 \\ 
$(1,\, 1,\, 0,\, 1,\, 0)$ & 4.309 & 4.352 & 4.134 & 3.611 & 3.055 & 4.132 & 3.585 & 2.973 \\ 
\rowcolor{mygray!80} $(0,\, 0,\, 0,\, 0,\,0)$ & 4.170 & 9.022 & 9.865 & 11.470 & 12.878 & 9.878 & 11.593 & 13.826 \\ 
\rowcolor{mygray!80} $(1,\, 1,\, 1,\, 1,\, 1)$ & 4.170 & 9.393 & 10.294 & 12.001 & 13.481 & 10.313 & 12.142 & 13.776 \\ 
$(0,\, 1, \, 0, \, 0, \,0)$ & 2.381 & 9.388 & 9.250 & 8.550 & 7.482 & 9.245 & 8.501 & 7.338 \\ 
$(1,\, 1,\, 1,\, 0,\, 1)$ & 2.381 & 9.797 & 9.679 & 8.985 & 7.892 & 9.677 & 8.939 & 7.749 \\ 
\rowcolor{mygray!80} $(1,\, 0,\, 0,\, 1,\,0)$ & 0.786 & 5.887 & 4.760 & 3.046 & 2.016 & 4.747 & 2.928 & 1.708 \\ 
\rowcolor{mygray!80} $(1,\, 0,\, 1,\, 1,\, 0)$ & 0.786 & 5.899 & 4.764 & 3.041 & 2.009 & 4.751 & 2.924 & 1.702 \\
$(0,\, 1, \, 1, \, 0, \,0)$ & 0.786 & 9.970 & 8.574 & 6.017 & 4.079 & 8.556 & 5.872 & 3.724 \\ 
$(1,\, 1,\, 0,\, 0,\, 1)$ & 0.786 & 10.388 & 8.964 & 6.323 & 4.300 & 8.948 & 6.175 & 3.935  \\ 
\rowcolor{mygray!80} $(0,\, 0, \, 1, \, 0, \,0)$ & 0.434 & 10.406 & 8.646 & 5.745 & 3.767 & 8.629 & 5.598 & 3.398 \\ 
\rowcolor{mygray!80} $(1,\, 1,\, 0,\, 1,\, 1)$ & 0.434 & 10.838 & 9.027 & 6.018 & 3.951 & 9.013 & 5.868 & 3.569 \\ 
$(0,\, 1, \, 0, \, 1, \, 0)$ &0.143 & 11.014 & 7.930 & 4.035 & 2.232 & 7.897 & 3.794 & 1.692 \\ 
$(1,\, 0,\, 1,\, 0,\, 1)$ & 0.143 & 11.367 & 8.222 & 4.214 & 2.333 & 8.191 & 3.970 & 1.787 \\ 
\rowcolor{mygray!80} $(0,\, 1, \, 1, \, 1, \,0)$ & 0.143 & 10.983 & 7.909 & 4.032 & 2.233 & 7.877 & 3.791 & 1.694 \\ 
\rowcolor{mygray!80} $(1,\, 0,\, 0,\, 0,\,1)$ & 0.143 & 11.324 & 8.204 & 4.219 & 2.339 & 8.173 & 3.976 & 1.794 \\ 
$(0,\, 0, \, 0, \, 1, \,0)$ & 0.079 & 11.430 & 7.959 & 3.853 & 2.102 & 7.927 & 3.601 & 1.536 \\ 
$(1,\, 0,\, 1,\, 1,\, 1)$ & 0.079  & 11.804 & 8.247 & 4.012 & 2.186 & 8.217 & 3.757 & 1.612 \\ 
\rowcolor{mygray!80} $(0,\, 0, \, 1, \, 1, \,0)$ & 0.026 & 12.293 & 7.442 & 2.816 & 1.457 & 7.395 & 2.495 & 0.780  \\  
\rowcolor{mygray!80} $(1,\, 0,\, 0,\, 1,\, 1)$ & 0.026 & 12.672 & 7.704 & 2.929 & 1.502 & 7.658 & 2.604 & 0.820 \\ 
$(0,\, 1, \, 1, \, 0, \,1)$ & 0.026 & 14.821 & 9.107 & 3.510 & 1.742 & 9.055 & 3.145 & 1.013  \\ 
$(0,\, 1, \, 0, \, 0, \,1)$ & 0.026 & 14.815 & 9.109 & 3.513 & 1.744 & 9.058 & 3.150 & 1.014  \\
\rowcolor{mygray!80} $(0,\, 0, \, 0, \, 0, \,1)$ & 0.014 & 15.092 & 9.030 & 3.368 & 1.681 & 8.979 & 2.993 & 0.934  \\ 
\rowcolor{mygray!80} $(0,\, 1, \, 1, \, 1, \,1)$ & 0.014 & 15.130 & 9.029 & 3.354 & 1.674 & 8.977 & 2.977 & 0.924 \\ 
$(0,\, 0, \, 1, \, 0, \,1)$ & 0.005 & 15.959 & 8.258 & 2.438 & 1.254 & 8.191 & 2.008 & 0.456 \\ 
$(0,\, 1, \, 0, \, 1, \,1)$ & 0.005 & 16.006 & 8.274 & 2.438 & 1.253 & 8.206 & 2.006 & 0.454 \\ 
\rowcolor{mygray!80} $(0,\, 0, \, 0, \, 1, \,1)$ & 0.001 & 17.194 & 7.487 & 1.758 & 1.035 & 7.403 & 1.276 & 0.204 \\ 
\rowcolor{mygray!80} $(0,\, 0, \, 1, \, 1, \,1)$ & 0.001  & 17.198 & 7.484 & 1.757 & 1.035 & 7.400 & 1.275 & 0.204 \\
\midrule
\multicolumn{2}{c|}{Gross-error sensitivity} & 17.198  & 10.294 & 12.001 & 13.481 & 10.313 & 12.142 & 13.826   \\
\bottomrule
\end{tabular}
\begin{flushleft}
\begingroup
\fontsize{9pt}{10.5pt}\selectfont
$\bullet \ \ \ $Response: Response patterns $\bm{u} \in \{0,1\}^5$.\\
$\bullet \ \ \ $Prob.: The probability of each response pattern occurring (expressed as a percentage).\\
$\bullet \ \ \ $MMLE: Marginal maximum likelihood estimation.\\
$\bullet \ \ \ $DPD and $\gamma$-divergence: The proposed methods (the second row of the table provides the hyperparameter values).
\endgroup
\end{flushleft}
\label{IF}
\end{table}

\section{Conclusion and Future Directions}\label{Conclusion and Future Directions}

This study proposed novel robust estimation methods via divergence measures to mitigate the impact of aberrant responses on item parameter estimation in IRT models. Unlike conventional marginal maximum likelihood estimation (MMLE), which is based on KL divergence, the proposed approaches introduce density power divergence (DPD) and $\gamma$-divergence as alternative measures. These divergences can control sensitivity to outliers through hyperparameter tuning, achieving both high robustness and efficiency when the parameters are appropriately selected. Since DPD and $\gamma$-divergence are generalizations of KL divergence, the proposed robust estimation methods can be considered extensions of standard MMLE.

The simulation results presented in Section \ref{Simulation} demonstrated that the proposed methods utilizing DPD and 
$\gamma$-divergence consistently outperformed conventional MMLE in estimation accuracy under various conditions with aberrant responses. Compared with robust MMLE proposed by Hong and Cheng (\cite{hong2019}), the proposed methods generally achieved better performance, although robust MMLE outperformed them under specific conditions. Even in scenarios with few aberrant responses, unlike robust MMLE, the proposed methods, with small hyperparameter settings, achieved estimation accuracy comparable to that of MMLE. An analysis based on influence functions verified that responses with low occurrence probabilities (i.e., those considered aberrant) exert less influence on the estimators as the hyperparameter values of the proposed methods increase. This result provides theoretical validation that larger hyperparameter values enhance robustness against aberrant responses. Overall, the findings highlight the robustness and versatility of the proposed methods, establishing them as a reliable alternative for item parameter estimation in IRT models across diverse conditions.

A significant difference between the proposed methods and standard MMLE pertains to the optimization process employed in the algorithms' update equations. While MMLE updates parameters for each item independently as implied by Eq.~\eqref{mmle_eq}, the proposed methods require simultaneous optimization of all item difficulty parameters. This joint optimization allows broader utilization of response information across items for each item's estimation, enhancing robustness to outliers. However, with increasing test length, the simultaneous optimization approach may lead to instability in parameter estimation and an increase in computational costs.

To extend the robust estimation of 1PLM item parameters to more general IRT models, improving estimation stability and computational efficiency is essential while maintaining robustness to outliers. One promising approach is to design algorithms that group multiple items and update parameters in stages rather than optimizing all item parameters simultaneously. Such strategies could maintain a certain level of robustness while ensuring stable estimation and practical computation times for large-scale parameter settings. Additionally, leveraging acceleration techniques for EM algorithms, such as those introduced by Ramsay (\cite{ramsay1975}) and Jamshidian and Jennrich (\cite{jamshidian1997}), could further improve the convergence speed of the algorithms. For instance, the R package `mirt' (Chalmers, \cite{chalmers2012}) employs Ramsay's acceleration method by default, enhancing computational efficiency. These insights suggest a promising direction for achieving a balance between robustness to outliers and computational efficiency, thereby broadening the applicability of IRT.

A key challenge in applying robust estimation methods based on DPD and $\gamma$-divergence lies in determining suitable hyperparameters. If these parameters are not properly set, the methods may fail to mitigate the influence of outliers effectively or may introduce excessive bias and variance. Recent studies (e.g., Basak et al., \cite{basak2021}; Sugasawa and Yonekura, \cite{sugasawa2021}) have advanced this field by introducing innovative frameworks for hyperparameter selection, paving the way for more structured and empirically grounded approaches to optimization. These advancements are crucial in facilitating a well-balanced trade-off between outlier resistance and overall estimation accuracy in robust estimation methods.

In this study, we proposed robust estimation methods based on DPD and $\gamma$-divergence as extensions of MMLE. We confirmed through simulation experiments that they consistently demonstrated stable estimation accuracy across diverse scenarios. On the other hand, challenges related to computational burden and hyperparameter tuning were also highlighted. Future research should explore extensions of 2PLM, 3PLM, polytomous models, and multidimensional models, while also developing fast algorithms for large-scale parameter estimation and establishing more systematic hyperparameter selection methods. By pursuing these directions, more versatile and practical robust estimation methods are expected to be developed, enabling their application in various testing scenarios.

\newpage
\renewcommand*{\bibfont}{\small}
\printbibliography

\newpage
\section*{Appendix}
\begin{appendix}
\section{The Explicit Expression of 
  \texorpdfstring{$V_{\beta}(\bm{b}; \bm{u})$, 
  $V_{\gamma}(\bm{b}; \bm{u})$, 
  and $V_{\text{KL}}(\bm{b}; \bm{u})$}{Explicit Expression}}
\label{appendixA}
In this appendix, we provide the explicit expressions of the Jacobian matrix $V_{\beta}(\bm{b}; \bm{u})$, $V_{\gamma}(\bm{b}; \bm{u})$ and $V_{\text{KL}}(\bm{b}; \bm{u})$ introduced in Section \ref{Asymptotic Properties}. Let $\Sigma(\theta,\bm{b})$ denote the Jacobian matrix of the function $\bm{b} \mapsto \xi(\bm{u},\theta,\bm{b})$. This matrix is a $J\times J$ diagonal matrix, with its $(j,j)$ entry given by 
\[
-D^2P_j(\theta)(1-P_j(\theta)).
\]
Then, $V_{\beta}(\bm{b}; \bm{u})$, $V_{\gamma}(\bm{b}; \bm{u})$ and $V_{\text{KL}}(\bm{b}; \bm{u})$ are respectively expressed as
\begin{align*}
V_{\beta}(\bm{b};\bm{u}) =& -\frac{\int q(\bm{u}|\theta;\bm{b})^{1+\beta} \{ (1+\beta)\xi(\bm{u},\theta,\bm{b}) \xi(\bm{u},\theta,\bm{b})^{\top} + \Sigma(\theta,\bm{b}) \} f(\theta)d\theta}{q(\bm{u};\bm{b})} \\
& + \frac{\int q(\bm{u}|\theta;\bm{b})^{1+\beta}\xi(\bm{u},\theta,\bm{b}) f(\theta) d\theta \int q(\bm{u}|\theta;\bm{b})\xi(\bm{u},\theta,\bm{b})^{\top} f(\theta) d\theta}{q(\bm{u};\bm{b})^2} \\
& + \int q(\bm{u}|\theta;\bm{b})^{1+\beta} \{ (1+\beta)\xi(\bm{u},\theta,\bm{b}) \xi(\bm{u},\theta,\bm{b})^{\top} + \Sigma(\theta,\bm{b}) \} f(\theta)d\theta d\bm{u}, \\[10pt]
V_{\gamma}(\bm{b};\bm{u}) =& -\frac{\int q(\bm{u}|\theta;\bm{b})^{1+\gamma} \{ (1+\gamma)\xi(\bm{u},\theta,\bm{b}) \xi(\bm{u},\theta,\bm{b})^{\top} + \Sigma(\theta,\bm{b}) \} f(\theta)d\theta}{q(\bm{u};\bm{b})} \int q(\bm{u}|\theta;\bm{b})^{1+\gamma} f(\theta) d\theta d\bm{u} \\
& + \frac{\int q(\bm{u}|\theta;\bm{b})^{1+\gamma}\xi(\bm{u},\theta,\bm{b}) f(\theta) d\theta \int q(\bm{u}|\theta;\bm{b})\xi(\bm{u},\theta,\bm{b})^{\top} f(\theta) d\theta}{q(\bm{u};\bm{b})^2}  \int q(\bm{u}|\theta;\bm{b})^{1+\gamma} f(\theta) d\theta d\bm{u} \\
& - \frac{\int q(\bm{u}|\theta;\bm{b})^{1+\gamma} f(\theta) d\theta \int q(\bm{u}|\theta;\bm{b})\xi(\bm{u},\theta,\bm{b})^{\top} f(\theta) d\theta}{q(\bm{u};\bm{b})^2} \int q(\bm{u}|\theta;\bm{b})^{1+\gamma} \xi(\bm{u},\theta,\bm{b}) f(\theta) d\theta d\bm{u} \\
& + \frac{\int q(\bm{u}|\theta;\bm{b})^{1+\gamma} f(\theta)d\theta}{q(\bm{u};\bm{b})} \int q(\bm{u}|\theta;\bm{b})^{1+\gamma} \{ (1+\gamma)\xi(\bm{u},\theta,\bm{b}) \xi(\bm{u},\theta,\bm{b})^{\top} + \Sigma(\theta,\bm{b})\} f(\theta)d\theta d\bm{u}, \\[10pt]
V_{\text{KL}}(\bm{b}; \bm{u}) = &-\frac{\int q(\bm{u}|\theta;\bm{b}) \{\xi(\bm{u},\theta,\bm{b}) \xi(\bm{u},\theta,\bm{b})^{\top} + \Sigma(\theta,\bm{b}) \} f(\theta)d\theta}{q(\bm{u};\bm{b})} \\
& + \frac{\int q(\bm{u}|\theta;\bm{b})\xi(\bm{u},\theta,\bm{b}) f(\theta) d\theta \int q(\bm{u}|\theta;\bm{b})\xi(\bm{u},\theta,\bm{b})^{\top} f(\theta) d\theta}{q(\bm{u};\bm{b})^2}.\\
\end{align*}
\end{appendix}

\end{document}